\documentclass[runningheads]{llncs}

\usepackage{newclude}
\usepackage{graphicx}
\usepackage{amsfonts}
\usepackage{enumerate}
\usepackage[scr]{rsfso}
\usepackage{units}
\usepackage{enumitem}
\usepackage{comment}
\usepackage{color}
\usepackage{textcomp}
\usepackage{hyperref}
\usepackage{multirow}
\usepackage{environ}
\usepackage[boxed]{algorithm2e}
\usepackage{environ}
\usepackage{comment}
\usepackage{makecell}
\usepackage[classfont=roman,funcfont=roman,langfont=roman]{complexity}
\usepackage{tikz}
\usepackage[thmmarks]{ntheorem} 
\usepackage{amsmath}

\newcommand{\MHVfull}{\textsc{Maximum Happy Vertices}}
\newcommand{\MHE}{\textsc{MHE}}
\newcommand{\MHV}{\textsc{MHV}}
\newcommand{\MHEfull}{\textsc{Maximum Happy Edges}}
\newtheorem{rrule}{Reduction rule}

\usepackage{tabularx,lipsum,environ,amsmath,amssymb}
\usepackage{mathtools}

\newif\ifshort

\ifshort
	\makeatletter
	\newenvironment{lemmaO}[1][]{
		\if\relax\detokenize{#1}\relax
		\expandafter\@firstoftwo
		\else
		\expandafter\@secondoftwo
		\fi
		{\begin{lemma}[$\star$]}{\begin{lemma}[#1, $\star$]}
	}
	{
		\end{lemma}
	}
	\makeatother
	\makeatletter
\newenvironment{theoremO}[1][]{
	\if\relax\detokenize{#1}\relax
	\expandafter\@firstoftwo
	\else
	\expandafter\@secondoftwo
	\fi
	{\begin{theorem}[$\star$]}{\begin{theorem}[#1, $\star$]}
		}
		{
		\end{theorem}
	}
	\makeatother
	\makeatletter
\newenvironment{corollaryO}[1][]{
	\if\relax\detokenize{#1}\relax
	\expandafter\@firstoftwo
	\else
	\expandafter\@secondoftwo
	\fi
	{\begin{corollary}[$\star$]}{\begin{corollary}[#1, $\star$]}
		}
		{
		\end{corollary}
	}
	\makeatother
	\NewEnviron{proofO}{}
\else

	\makeatletter
	\newenvironment{proofO}[1][]{\begin{proof}}{\end{proof}}
	\makeatother
	
\fi

\makeatletter
\newcommand{\problemtitle}[1]{\gdef\@problemtitle{#1}}\newcommand{\probleminput}[1]{\gdef\@probleminput{#1}}\newcommand{\problemquestion}[1]{\gdef\@problemquestion{#1}}\NewEnviron{problemx}{
	\problemtitle{}\probleminput{}\problemquestion{}\BODY \par \noindent
	\centering
	\boxed
	{
		\begin{tabularx}{340pt}{@{\hspace{0.1cm}} l p{285pt} c}
			\multicolumn{2}{@{\hspace{0.1cm}}l}{\@problemtitle} \\\textbf{Input:} & \@probleminput \\\textbf{Question:} & \@problemquestion \end{tabularx}}
	\par\addvspace{.5\baselineskip}
}

\renewcommand{\O}{\mathcal{O}}
\newcommand{\Ostar}[1]{\O^*(#1)}

\theoremseparator{.}
\newtheorem{theorem}{Theorem}
\newtheorem{lemma}{Lemma}
\newtheorem{corollary}{Corollary}

\theorembodyfont{\normalfont}
\theoremseparator{.}
\newtheorem{definition}{Definition}
\newtheorem{claim}{Claim}

\theoremstyle{nonumberplain}
\theoremheaderfont{\itshape}
\theoremsymbol{\ensuremath{\Box}}
\theorembodyfont{\normalfont}
\newtheorem{proof}{Proof}

\makeatletter
\newtheoremstyle{nonumberplainnobrackets}{\item[\theorem@headerfont\hskip\labelsep ##1\theorem@separator]}{\item[\theorem@headerfont\hskip \labelsep ##1\ ##3\theorem@separator]}
\makeatother

\theoremstyle{nonumberplainnobrackets}
\theoremsymbol{\ensuremath{\blacksquare}} 
\theorembodyfont{\normalfont}
\theoremheaderfont{\itshape}
\newtheorem{claimproof}{Proof of Claim \theclaim}

\ifshort
\title{Lower Bounds for the Happy Coloring Problems\thanks{This research was supported by the Russian Science Foundation (project 16-11-10123)}}
\else
\title{Lower Bounds for the Happy Coloring Problems\thanks{This research was supported by the Russian Science Foundation (project 16-11-10123)}}
\fi
\author{Ivan Bliznets\inst{1,2} \and
	Danil Sagunov\inst{1}}
\authorrunning{I. Bliznets and D. Sagunov}
\institute{St.\ Petersburg Department of Steklov Institute of Mathematics of the Russian Academy
	of Sciences, Saint Petersburg, Russia\\\email{iabliznets@gmail.com}, \email{danilka.pro@gmail.com} \and National Research University Higher School of Economics, Saint Petersburg, Russia}

\begin{document}
	\maketitle

\begin{abstract}
In this paper, we study the \textsc{Maximum Happy Vertices} and the \textsc{Maximum Happy Edges} problems (MHV and MHE for short).
Very recently, the problems attracted a lot of attention and were studied in Agrawal '17, Aravind et al.\ '16, Choudhari and Reddy '18, Misra and Reddy '17.
Main focus of our work is lower bounds on the computational complexity of these problems.
Established lower bounds can be divided into the following groups: \NP-hardness of the above guarantee parameterization, kernelization lower bounds (answering questions of Misra and Reddy '17), exponential lower bounds under the \textsc{Set Cover Conjecture} and the \textsc{Exponential Time Hypothesis}, and inapproximability results.
Moreover, we present an $\O^*(\ell^k)$ randomized algorithm for MHV and an $\O^*(2^k)$ algorithm for MHE, where $\ell$ is the number of colors used and $k$ is the number of required happy vertices or edges.
These algorithms cannot be improved to subexponential taking proved lower bounds into account.

\end{abstract}

\section{Introduction}

In this paper, we study \textsc{Maximum Happy Vertices} and \textsc{Maximum Happy Edges}. The problems are motivated by a study of algorithmic aspects of homophyly law in large networks and  were  introduced by Zhang and Li in 2015~\cite{zhang2015algorithmic}.
The law states that in social networks people are more likely to connect with people they like.
Social network is represented by a graph, where each vertex corresponds to a person of the network, and an edge between two vertices denotes that the corresponding persons are connected within the network.
Furthermore, we let vertices have a color assigned.
The color of a vertex indicates type, character or affiliation of the corresponding person in the network.
An edge is called \emph{happy} if its endpoints are colored with the same color.
A vertex is called \emph{happy} if all its neighbours are colored with the same color as the vertex itself.
Equivalently, a vertex is happy if all edges incident to it are happy.
Formal definition of \textsc{Maximum Happy Vertices} and \textsc{Maximum Happy Edges} is the following.

\begin{problemx}
	\problemtitle{\MHVfull~(\MHV)}
	\probleminput{A graph $G$, a partial coloring of vertices $p: S \rightarrow [\ell]$ for some $S\subseteq V(G)$ and an integer $k$.}
	\problemquestion{Is there a coloring $c: V(G) \rightarrow [\ell]$ extending partial coloring $p$ such that the number of happy vertices with respect to $c$ is at least $k$?}
\end{problemx}

\begin{problemx}
	\problemtitle{\MHEfull~(\MHE)}
	\probleminput{A graph $G$, a partial coloring of vertices $p: S \rightarrow [\ell]$ for some $S\subseteq V(G)$ and an integer $k$.}
	\problemquestion{Is there a coloring $c: V(G) \rightarrow [\ell]$ extending partial coloring $p$ such that the number of happy edges with respect to $c$ is at least $k$?}
\end{problemx}

Recently, \MHV~and \MHE~have attracted a lot of attention and were studied from parameterized~\cite{Agrawal2018,Aravind2016,Aravind2017,Choudhari2018,Misra2018} and  approximation~\cite{zhang2015algorithmic,zhang2018improved, zhang2015improved,xu2016submodular} points of view as well as from experimental perspective~\cite{lewis2019finding}.

NP-hardness of \MHV and MHE was proved by Zhang and Li even in case when only three colors are used. Later,  Misra and Reddy \cite{Misra2018} proved \NP-hardness of both \MHV~and \MHE~on split and on bipartite graphs. However, \MHV~is polynomially time solvable on cographs and trees~\cite{Misra2018,Aravind2016}. Approximation results for \MHV~are presented in Zhang et al.~\cite{zhang2018improved}. They showed that \MHV~can be approximated within $\frac{1}{\Delta + 1}$, where $\Delta$ is the maximum degree of the input graph, and \MHE~can be approximated within $\frac{1}{2}+\frac{\sqrt{2}}{4}f(\ell)$, where $f(\ell)= \frac{(1-1/\ell)\sqrt{\ell(\ell-1)}+1/\sqrt{2}}{\ell-1+1/2\ell}$. From parameterized point of view the following parameters were studied: pathwidth~\cite{Agrawal2018,Aravind2017}, treewidth~\cite{Agrawal2018,Aravind2017}, neighbourhood diversity~\cite{Aravind2017}, vertex cover~\cite{Misra2018}, distance to clique~\cite{Misra2018}, distance to threshold graphs~\cite{Choudhari2018}. Kernelization questions were studied in works \cite{Agrawal2018,gao2018kernelization}.
Agrawal~\cite{Agrawal2018} provided a $\O(k^2\ell^2)$ kernel for MHV where $\ell$ is the number of colors used and $k$ is the number of desired happy vertices.
Independently, Gao and Gao~\cite{gao2018kernelization} present $2^{k \ell + k} + k\ell +k+\ell$ kernel for general case and $7(k\ell+k)+\ell-10$ in case of planar graphs.

Short summary of our results can be found below.

\begin{description}
	\item [\textbf{No polynomial kernels:}] If $\NP \not \subseteq \coNP/\poly$ then there are no polynomial kernels for MHV parameterized by vertex cover, and no polynomial kernels for MHE under the following parametrizations: number of uncolored vertices, number of happy edges, and distance to almost any reasonable graph class.  Moreover, under $\NP \not\subseteq \coNP/\poly$, there is no $\O( (k^d \ell)^{2-\epsilon} )$ and no $\O( (k^dh)^{2-\epsilon} )$  bitsize kernel for \MHV. Note that these results answer question from~\cite{Misra2018}:
	"Do the \MHVfull~and \MHEfull~problems admit polynomial kernels when parameterized by either the vertex
	cover or the distance to clique parameters?"
	\item [\textbf{Above guarantee:}] Above-greedy versions of \MHV~and \MHE~are NP-complete even for budget equal $1$.
	\item [\textbf{Exponential lower bounds:}] Assuming the Set Cover Conjecture, \MHV~and \MHE~do not admit $\O^*((2-\epsilon)^{n'})$ algorithms, where $n'$ is the number of uncolored vertices in the input graph. Even with $\ell=3$, there is no $2^{o(n+m)}$ algorithm for \MHV~and \MHE, unless ETH fails. 
	\item [\textbf{Innaproximability:}] Unless $\P= \NP$, \MHV~does not admit approximation algorithm with factors $\O(n^{\frac{1}{2}-\epsilon})$,  $\O(m^{\frac{1}{2}-\epsilon})$,  $\O(h^{1-\epsilon})$,  $\O({\ell}^{1-\epsilon})$, for any $\epsilon > 0$.
	
	\item [\textbf{Algorithms:}] We present $\O^*(\ell^k)$ randomized algorithm for MHV and $\O^*(2^k)$ algorithm for MHE. Running time of this algorithms match with the corresponding lower bounds.
	We should note that an algorithm with the running time of $\O^*(2^k)$ for \MHE~was also presented by Aravind et al.\ in \cite{Aravind2017}.
\end{description}

\section{Preliminaries}

\textbf{Basic notation.} We denote the set of positive integer numbers by $\mathbb{N}$.
For each positive integer $k$, by $[k]$ we denote the set of all positive integers not exceeding $k$, $\{1,2,\ldots, k\}$.
We use $\sqcup$ for the disjoint union operator, i.e.\ $A\sqcup B$ equals $A\cup B$, with an additional constraint that $A$ and $B$ are disjoint.

We use traditional $\O$-notation for asymptotical upper bounds.
We additionally use $\mathcal{O}^*$-notation that hides polynomial factors.
Many of our results concern the parameterized complexity of the problems, including fixed-parameter tractable algorithms, kernelization algorithms, and some hardness results for certain parameters.
For detailed survey in parameterized algorithms we refer to the book of Cygan et al.\ \cite{cygan2015parameterized}.

Throughout the paper, we use standard graph notation and terminology, following the book of Diestel \cite{diestel2018graph}.
All graphs in our work are undirected simple graphs.
We may refer to the distance to $\mathcal{G}$ parameter, where $\mathcal{G}$ is an arbitrary graph class.
For a graph $G$, we say that a vertex subset $S\subseteq V(G)$ is a \emph{$\mathcal{G}$ modulator} of $G$, if $G$ becomes a member of $\mathcal{G}$ after deletion of $S$, i.e.\ $G\setminus S \in \mathcal{G}$.
Then, the \emph{distance to $\mathcal{G}$} parameter of $G$ is defined as the size of its smallest $\mathcal{G}$ modulator.

\textbf{Graph colorings.} When dealing with instances of \MHVfull~or \MHEfull, we use a notion of colorings.
A \textit{coloring} of a graph $G$ is a function that maps vertices of the graph to the set of colors.
If this function is partial, we call such coloring \emph{partial}.
If not stated otherwise, we use $\ell$ for the number of distinct colors, and assume that colors are integers in $[\ell]$.
A partial coloring $p$ is always given as a part of the input for both problems, along with graph $G$.
We also call $p$ a \emph{precoloring} of the graph $G$, and use $(G,p)$ to denote the graph along with the precoloring.
The goal of both problems is to extend this partial coloring to a specific coloring $c$ that maps each vertex to a color.
We call $c$ a \emph{full coloring} (or simply, a coloring) of $G$ that extends $p$.
We may also say that $c$ is a coloring of $(G,p)$.
For convenience, introduce the notion of potentially happy vertices, both for full and partial colorings.

\begin{definition}
	We call a vertex $v$ of $(G,p)$ \textit{potentially happy}, if there exists a coloring $c$ of $(G,p)$ such that $v$ is happy with respect to $c$. In other words, if $u$ and $w$ are precolored neighbours of $v$, then $p(u)=p(w)$. We denote the set of all potentially happy vertices in $(G,p)$ by $\mathcal{H}(G,p)$.
	
	By $\mathcal{H}_i(G,p)$ we denote the set of all potentially happy vertices in $(G,p)$ such that they are either precolored with color $i$ or have a neighbour precolored with color $i$: $$\mathcal{H}_i(G,p)=\{v \in \mathcal{H}(G,p) \mid N[v] \cap p^{-1}(i)\neq \emptyset\}.$$
	In other words, if a vertex $v \in \mathcal{H}_i(G,p)$ is happy with respect to some coloring $c$ of $(G,p)$, then necessarily $c(v)=i$.
\end{definition}

For a graph with precoloring $(G,p)$, by $h=|\mathcal{H}(G,p)|$ we denote the number of potentially happy vertices in $(G,p)$.
Note that if $c$ is a full coloring of a graph $G$, then $|\mathcal{H}(G,c)|$ is equal to the number of vertices in $G$ that are happy with respect to $c$.

\section{Polynomial kernels for structural graph parameters}\label{sec:structural-kernel}

In this section, we study existence of polynomial kernels for \MHV~or \MHE~under several parameterizations.
We start with proving lower bounds for structural graph parameters.
We provide reductions to both \MHV~and \MHE~from the following problem.

\begin{problemx}
	\problemtitle{\textsc{Bounded Rank Disjoint Sets} \cite{Dom2014}}
	\probleminput{A set family $\mathcal{F}$ over a universe $U$ with every set $S \in \mathcal{F}$ having size at most $d$,
		and a positive integer $k$.}
	\problemquestion{Is there a subfamily $\mathcal{F}'$ of $\mathcal{F}$ of size at least $k$ such that every pair of sets
		$S_1$, $S_2$ in $\mathcal{F}'$ we have $S_1 \cap S_2 = \emptyset$?}
\end{problemx}

\begin{theorem}[\cite{Dom2014}]\label{thm:disj_sets_no_polycomp}
	\textsc{Bounded Rank Disjoint Sets} parameterized by $kd$ does not admit a polynomial compression even if every set $S \in \mathcal{F}$ consists of exactly $d$ elements and $|U|=kd$, unless $\NP \subseteq \coNP/\poly$.
\end{theorem}

The following two theorems answer open questions posed in \cite{Misra2018}.

\begin{theorem}\label{thm:mhv_nopolycomp}
	\MHVfull~parameterized by the vertex cover number does not admit a polynomial compression, unless $\NP \subseteq \coNP/\poly$.
\end{theorem}
\begin{proof}
	We give a polynomial reduction from the \textsc{Set Packing} problem, such that the vertex cover number of the constructed instance of \MHV~is at most the size of the universe of the initial instance of \textsc{Set Packing} plus one. Since \textsc{Bounded Rank Disjoint Sets} is a special case of \textsc{Set Packing}, from Theorem \ref{thm:disj_sets_no_polycomp} the theorem statement will then follow. The reduction is as follows.
	
	Given an instance $(U=[n], \mathcal{F}=\{S_1, S_2, \ldots, S_m\}, k)$ of \textsc{Set Packing}, construct an instance $(G, p, k)$ of \textsc{MHV}. For each $i \in U$, introduce vertex $u_i$ in $G$ and left it uncolored. For each set $S_j \in \mathcal{F}$, introduce a vertex $s_j$ in $G$ and precolor it with color $j$, i.e.\ $p(s_j)=j$. Thus, the set of colors used in precoloring $p$ is exactly $[m]$. Then, for each $i \in [n]$ and $j \in [m]$ such that $i \in S_j$, introduce an edge between $u_i$ and $s_j$ in $G$. Additionally, introduce two vertices $t_1$ and $t_2$ to $G$ and precolor them with colors $1$ and $2$ respectively. Then, introduce an edge $(t_1, t_2)$ to $G$ and for every $i \in [n]$ and $j \in [2]$, introduce an edge $(u_i,t_j)$ in $G$. Thus, vertices $t_1$ and $t_2$ never become happy and ensure that $u_i$ never become happy for any $i \in [n]$. Finally, set the number of required happy vertices to $k$. Observe that $\{u_1, \ldots, u_n\} \cup \{t_1\}$ forms a vertex cover of $G$, hence the vertex cover number of $G$ is at most $n+1$.
	
	We now claim that $(U, \mathcal{F}, k)$ is a yes-instance of \textsc{Set Packing} if and only if $(G, p, k)$ is a yes-instance of \MHV. Let $S_{i_1}, S_{i_2}, \ldots, S_{i_k}$ be the answer to $(U, \mathcal{F}, k)$, i.e.\ $S_{i_p}\cap S_{i_q}=\emptyset$ for every distinct $p, q \in [k]$. Since $S_{i_1}, S_{i_2}, \ldots, S_{i_k}$ are disjoint, $s_{i_1}, s_{i_2}, \ldots, s_{i_k}$ do not have any common neighbours in $G$. Hence, we can extend coloring $p$ to coloring $c$ in a way that $s_{i_1}, s_{i_2}, \ldots, s_{i_k}$ are happy with respect to $c$ ($c(u_i)$ is then, in fact, the index of the set containing $u_i$, i.e.\ $u_i \in S_{c(u_i)}$). At least $k$ vertices become happy in $G$, hence $(G, p, k)$ is a yes-instance of \MHV.
	
	In the other direction, let $c$ be a coloring of $G$ extending $p$ so that at least $k$ vertices in $G$ are happy with respect to $c$. Only vertices that can be happy in $(G,p)$ are vertices of type $s_i$, hence there are vertices $s_{i_1}, s_{i_2}, s_{i_3}, \ldots, s_{i_k}$ that are happy in $G$ with respect to $c$. Since these vertices are precolored with pairwise distinct colors and are simultaneously happy, they may have no common neighbours in $G$. This implies that the corresponding sets of the initial instance $S_{i_1}, S_{i_2}, \ldots, S_{i_k}$ are pairwise disjoint. Hence, they form an answer to the initial instance $(U, \mathcal{F}, k)$ of \textsc{Set Packing}. This completes the proof.
\end{proof}

\begin{theorem}\label{thm:mhe_nopolycomp}
	\MHEfull~parameterized by the number of uncolored vertices or by the number of happy edges does not admit a polynomial compression, unless $\NP \subseteq \coNP/\poly$.
\end{theorem}

\begin{proof}
	As in the proof of Theorem \ref{thm:mhv_nopolycomp}, we again provide a polynomial reduction from \textsc{Bounded Rank Disjoint Sets} and then use Theorem \ref{thm:disj_sets_no_polycomp}.
	In this proof though, we will use the restricted version \textsc{Bounded Rank Disjoints Sets} problem itself (and not the \textsc{Set Packing} problem), formulated in Theorem \ref{thm:disj_sets_no_polycomp}.
	That is, we will use the constraint that all sets in the given instance are of the same size $d$, and the size of the universe $|U|$ is equal to $kd$.
	We note that the following reduction has very much in common with the reduction described in the proof of Theorem \ref{thm:mhv_nopolycomp}.
	
	Given an instance $([n], \mathcal{F}=\{S_1, S_2, \ldots, S_m\}, k)$ of \textsc{Bounded Rank Disjoint Sets} with $n=kd$ and $|S_i|=d$ for every $i\in [m]$, we construct an instance $(G, p, k')$ of \MHE. We assume that each element of the universe $[n]$ is contained in at least one set, otherwise the given instance is a no-instance. Firstly, as in the proof of Theorem \ref{thm:mhv_nopolycomp}, for each element of the universe $i \in [n]$, introduce a corresponding vertex $u_i$ in $G$.
	For each set $S_j$, $j \in [m]$, introduce not just one, but $n$ corresponding vertices $s_{j,1}, s_{j,2}, \ldots, s_{j,n}$.
	Then again, similarly to the proof of Theorem \ref{thm:disj_sets_no_polycomp}, for each $i, j$ such that $i \in S_j$, introduce edges between $u_i$ and \emph{each} vertex $s_{j, t}$ corresponding to the set $S_j$, i.e.\ $n$ edges in total.
	To finish the construction of $G$, introduce every possible edge $(u_i, u_j)$ in $G$.

	Thus, $V(G)=\{u_i \mid i \in [n]\} \cup \{s_{j,t} \mid j \in [m], t \in [n]\}$ and $E(G)=\{(u_i, s_{j,t}) \mid i \in S_j, t \in [n] \cup \{(u_i, u_j) \mid i, j \in [n], i \neq j\}$. Then, precolor the vertices of $G$ in the usual way, i.e.\ set $p(s_{j,t})=j$ for every $j \in [m]$ and $t \in [n]$, and leave each vertex $u_i$ uncolored. Finally, we set the number of required happy edges to $k'=n^2+k\binom{d}{2}=(kd)^2+k\binom{d}{2}$. Construction of $(G, p, k')$ is done in polynomial time. Observe that the number of uncolored vertices in $(G, p, k)$ equals the size of the universe $n$, and the number of required happy edges is polynomial of $n$. Hence, existence of a polynomial kernel respectively to any of these two parameters for \MHE~contradicts the statement of Theorem \ref{thm:disj_sets_no_polycomp}. We argue that the initial instance is a yes-instance if and only if $(G, p, k')$ is a yes-instance of \MHE.
	
	We prove first that if $([n], \mathcal{F}, k)$ is a yes-instance, then $(G, p, k')$ is a yes-instance. Let $([n], \mathcal{F}, k)$ be a yes-instance of the restricted version of \textsc{Bounded Rank Disjoint Sets}, and let $S_{i_1}, S_{i_2}, \ldots, S_{i_k}$ be the instance solution. As usual, extend $p$ to a coloring $c$ of $G$ by setting $c(u_i)$ to the index of the set in the solution containing $u_i$, i.e.\ $c(u_i)=i_t$ for some $t \in k$ and $i \in S_{c(u_i)}$. Since $S_{i_1}, S_{i_2}, \ldots, S_{i_k}$ are disjoint, and their total size equals the size of the universe, such coloring $c$ always exists uniquely for a fixed solution of $([n], \mathcal{F}, k)$. We claim that there are exactly $k'$ happy edges in $G$ with respect to $c$.
	
	All edges in $G$ are either of type $(u_i, s_{j,t})$ or of type $(u_i, u_j)$. Consider edges of type $(u_i, s_{j,t})$ for a fixed $i \in [n]$. Happy edges among them are those with $c(s_{j,t})=c(u_i)$. Since $c(s_{j, t})=p(s_{j,t})=j$ and $i \in S_{c(u_i)}$, these edges are exacly $(u_i, s_{c(u_i), t})$. Hence, there are $n$ happy edges of this type for a fixed $i \in [n]$ and $n^2$ happy edges of this type in total. It is left to count the number of happy edges of the clique, i.e.\ edges of type $(u_i, u_j)$. Observe that each $u_i$ is colored with a color corresponding to a containing set of the answer. Since each set is of size $d$, the vertices $u_i$ are split by color into $k$ groups of size $d$. Each group contributes exactly $\binom{d}{2}$ happy edges, and no edge connecting vertices from different groups is happy. Thus, there are exactly $k\binom{d}{2}$ happy edges of type $(u_i, u_j)$ in $G$ with respect to $c$. We get that exactly $n^2+k\binom{d}{2}$ edges of $G$ are happy with respect to $c$, hence $(G, p, k')$ is a yes-instance of \MHE.
	
	In the other direction, let $(G, p, k')$ be a yes-instance of \MHE, and let $c$ be an optimal coloring of $G$ extending $p$. At least $k'$ edges are happy in $G$ with respect to $c$. Let us show that \emph{exactly} $k'$ edges are happy in $G$ with respect to $c$.
	\begin{claim}\label{claim:mhe_polycomp_req_cols}
		In any optimal coloring $c$ of $G$ extending $p$, $i \in S_{c(u_i)}$ for each $i \in [n]$.
	\end{claim}
	\begin{claimproof}
		Suppose it is not true, and $c$ is an optimal coloring of $(G,p)$ and $i \notin S_{c(u_i)}$ for some $i \in [n]$.
		For each $j$ with $i \in S_j$, $u_i$ is adjacent to $n$ vertices $s_{j,t}$, which are precolored with color $j$.
		None of edges $(u_i, s_{j,t})$ are happy with respect to $c$, since $j \neq c(u_i)$.
		The only other edges incident to $u_i$ are $n-1$ edges of the clique.
		Thus, $u_i$ is incident to at most $n-1$ happy edges.
		
		Choose arbitrary $j$ with $i \in S_j$, and put $c(u_i)=j$.
		$u_i$ becomes incident to at least $n$ happy edges.
		Happiness of edges not incident with $u_i$ has not changed.
		Thus, the change yields at least one more happy edge.
		A contradiction with the optimality of $c$.
	\end{claimproof}

	\begin{claim}\label{claim:mhe_polycomp_large_groups}
		In any optimal coloring $c$ of $G$ extending $p$, there are at most $k \binom{d}{2}$ happy edges of type $(u_i, u_j)$ in $G$ with respect to $c$.
	\end{claim}
	\begin{claimproof}
		The vertices $u_i$ are split into groups containing vertices of the same color by $c$, so the happy edges of type $(u_i, u_j)$ are exactly the edges inside the groups. By Claim \ref{claim:mhe_polycomp_req_cols}, each $u_i$ is colored with a color corresponding to a set containing $i$ in $c$. Hence, each group contains vertices corresponding to elements of the same set, and thus contains at most $d$ vertices. So each $u_i$ is incident to at most $d-1$ happy edges of type $(u_i, u_j)$, and in total there are at most $n \cdot (d-1)/2=k\binom{d}{2}$ such happy edges in $G$ with respect to $c$.
	\end{claimproof}

	From Claims \ref{claim:mhe_polycomp_req_cols} and \ref{claim:mhe_polycomp_large_groups} follows that at most $n^2+k\binom{d}{2}=k'$ edges are happy in $G$ with respect to $c$. And as seen in the proof of Claim \ref{claim:mhe_polycomp_large_groups}, the only way that yields exactly $k'$ happy edges is when $u_i$ are split by color into disjoint groups of size $d$, each containing vertices corresponding to a set of the initial instance. Hence, if $c$ yields $k'$ happy edges in $G$, $\{S_{c_(u_i)} \mid i \in [n]\}$ is a solution to $([n], \mathcal{F}, k)$. Thus, $([n], \mathcal{F}, k)$ is a yes-instance of \textsc{Bounded Rank Disjoint Sets}. This finishes the whole proof.
\end{proof}

\begin{definition}
	We call a graph family $\mathcal{G}$ \emph{uniformly polynomially instantiable}, if there is an algorithm that, given positive integer $n$ as input, outputs a graph $G$, such that $|V(G)| \ge n$ and $G\in \mathcal{G}$, in $\poly(n)$ time.
\end{definition}

\begin{corollaryO}
	For any uniformly polynomially instantiable graph family $\mathcal{G}$, \MHEfull, parameterized by the distance to graphs in $\mathcal{G}$, does not admit a polynomial compression, unless $\NP \subseteq \coNP/\poly$.
\end{corollaryO}
\begin{proofO}
	Suppose it is not true and there is a uniformly polynomially instantiable graph family $\mathcal{G}$, such that \MHE~parameterized by the distance to graphs in $\mathcal{G}$ admits a polynomial compression. We show how to reduce an instance of \MHE~with $d$ uncolored vertices to an instance of \MHE~with the distance to graph in $\mathcal{G}$ being at most $d$, and then get a contradiction with Theorem \ref{thm:mhe_nopolycomp}.
	
	Let $(G,p,k)$ be an instance of \MHE~with $d$ uncolored vertices. Denote the set of all uncolored vertices in $(G,p)$ by $U$ and the set of all precolored vertices by $P$, so $U\sqcup P=V(G)$. Assume that $G$ has no edge between vertices in $P$, otherwise delete it and decrease $k$ by one if its endpoints are of the same color in $(G, p)$. Construct an instance $(G',p',k')$ as follows. Use the algorithm that output instances of $\mathcal{G}$, with $|P|$ as input. The algorithm gives a graph $G''$, such that $|V(G'')|\ge |P|$ and $G'' \in \mathcal{G}$. Take an arbitrary subset $P' \subseteq V(G'')$ of size $|P|$, and identify its vertices with vertices in $P$. Construct $G'$ by introducing $|U|$ new vertices to $G''$,  that are identified with the vertices of $U$. Denote the set of these vertices by $U'$. Then, add an edge between vertices in $U'$ or between vertices in $U'$ and $P'$ if there is an edge between corresponding vertices in $G$. Finally, construct $p'$ by precoloring vertices in $P'$ with the color of their corresponding vertices in $(G,p)$, leave the vertices of $U'$ uncolored, and precolor all remaining vertices arbitrarily. There may be some happy edges between precolored vertices in $(G', p')$, let their number be $h$. Set $k'=k+h$.
	
	One may easily show that the constructed instance $(G', p', k')$ is a yes-instanse of \MHE~if and only if the initial instance $(G,p,k)$ is a yes-instance of \MHE. Moreover, deletion of $U'$ from $G'$ yields $G'' \in \mathcal{G}$. Hence, $G'$ has the distance to graphs in $\mathcal{G}$ being at most $|U'|=|U|=d$. We therefore obtain the required polynomial reduction that leads to the desired contradiction.
\end{proofO}

In the rest of the section we study kernel bitsize lower bounds for \MHV, parameterized by either $k+\ell$ or $k+h$, where $h$ is the number of potentially happy vertices.
This relates to the result of Agrawal in \cite{Agrawal2018}, where the author showed that \MHV~admits a polynomial kernel with $\O(k^2\ell^2)$ vertices.
We show that, for any $d>0$ and any $\epsilon>0$, there is no kernel of bitsize $\O(k^d \cdot \ell^{2-\epsilon})$ for \MHV.
Similarly, we show that there is no kernel of bitsize $\O(k^d \cdot h^{2-\epsilon})$ for \MHV.
To prove these lower bounds, we refer to the framework of weak cross-compositions, that originates from works of Dell and van Mekelbeek \cite{DellMekelbeek2014}, Dell and Marx \cite{Dell2012} and Hermelin and Wu \cite{Hermelin2012}.
These results are finely summarized by Cygan et al.\ in the chapter on lower bounds for kernelization \cite{Cygan2015}.
We recall the notion of weak cross-compositions.
\iffalse
Due to the lack of space, we omit the notion here.
It can be found in \cite{Cygan2015, Dell2012, Hermelin2012} or in the full version of our paper.
\else

\begin{definition}[\cite{Cygan2015, Dell2012, Hermelin2012}]
	Let $L\subseteq \Sigma^*$ be a language and $Q \subseteq \Sigma^* \times \mathbb{N}$ be a parameterized
	language.
	We say that $L$ \emph{weakly-cross-composes} into $Q$ if there
	exists a real constant $d \ge 1$, called the dimension, a polynomial equivalence
	relation $\mathcal{R}$, and an algorithm $\mathcal{A}$, called the \emph{weak cross-composition}, satisfying
	the following conditions.
	The algorithm $\mathcal{A}$ takes as input a sequence of
	$x_1, x_2, \ldots, x_t \in \Sigma^*$
	that are equivalent with respect to $\mathcal{R}$, runs in time
	polynomial in $\sum^t_{i=1} |x_i|$, and outputs one instance $(y, k) \in \Sigma^* \times \mathbb{N}$ such that:
	\begin{enumerate}[label=(\alph*)]
		\item\label{enum:weak_cc_poly_bound} for every $\delta > 0$ there exists a polynomial $p(\cdot)$ such that for every choice of $t$ and input strings $x_1, x_2, \ldots, x_t$ it holds that $k \le p(\max^t_{i=1} |x_i|) \cdot t^{\frac{1}{d}+\delta}$, and
		\item\label{enum:weak_cc_equiv} $(y, k) \in Q$ if and only if there exists at least one index i such that $x_i \in L$.
	\end{enumerate}
\end{definition}
\fi

The framework of weak cross-compositions is used for proving conditional lower bounds on polynomial compression bitsize.
This is formulated in the following theorem.

\begin{theorem}[\cite{Cygan2015, Dell2012, Hermelin2012}]\label{thm:book_weak_cc_bitsize_lower_bound}
	If an \NP-hard language $L$ admits a weak cross-composition of dimension $d$ into a parameterized language $Q$.
	Then for any $\epsilon > 0$, $Q$ does not admit a polynomial compression with bitsize $\O(k^{d-\epsilon})$, unless $\NP \subseteq \coNP/\poly$.
\end{theorem}

Dell and Marx \cite{Dell2012} use this framework to show that the \textsc{Vertex Cover} problem parameterized by the solution size does not admit a kernel with subquadratic bitsize.
Their result is the following.

\begin{lemma}[\cite{Dell2012, Cygan2015}]\label{lemma:book_weak_cross_comp_vc}
	There exists a weak cross composition of dimension $2$ from an \NP-hard problem \textsc{Multicolored Biclique} into the \textsc{Vertex Cover} problem parameterized by the solution size.
	In fact, this weak cross-composition $\mathcal{A}$, given instances $x_1, x_2, \ldots, x_t$ of \textsc{Multicolored Biclique} as input, outputs an instance $(G, k')$ of \textsc{Vertex Cover} satisfying
	\begin{itemize}
		\item $|V(G)| \le p(\max^t_{i=1} |x_i|) \cdot \sqrt{t}$, and
		\item $|V(G)|-k' \le q(\max^t_{i=1}|x_i|),$
	\end{itemize}
	for some polynomials $p$ and $q$.
\end{lemma}

The bound for $|V(G)|-k'$ is given because one can look at an instance $(G,k')$ of \textsc{Vertex Cover} as at an instance $(G, |V(G)|-k')$ of \textsc{Independent Set}.
Then, the solution parameter of \textsc{Independent Set} is bounded with polynomial of the maximum input size, independently of the number of instances $t$.
We are ready to prove the theorem.

\begin{theoremO}\label{thm:mhv_kernel_bitsize}
	For any fixed constant $d$ and any $\epsilon > 0$, \MHVfull~ does not admit polynomial compressions with bitsizes $\O((k^{d} \cdot \ell)^{2-\epsilon})$ and $\O((k^{d} \cdot h)^{2-\epsilon})$, where $h$ is the number of potentially happy vertices, unless $\NP \subseteq \coNP/\poly$.
\end{theoremO}
\begin{proofO}
	Let $d$ be an arbitrary fixed constant.
	We show that \textsc{Multicolored Biclique} admits a polynomial compression into \MHV~parameterized either by $k^d \cdot \ell$ or by $k^d \cdot h$.
	By Theorem \ref{thm:book_weak_cc_bitsize_lower_bound}, it is sufficient for proving the theorem.
	
	We extend the weak cross-composition $\mathcal{A}$ into \textsc{Vertex Cover} from Lemma \ref{lemma:book_weak_cross_comp_vc}.
	Thus, we obtain the desired weak cross-composition $\mathcal{A'}$ into \MHV.
	Firstly, $\mathcal{A'}$ runs $\mathcal{A}$ to obtain an instance $(G, k')$ of \textsc{Vertex Cover}.
	Equivalently, $(G, |V(G)|-k')$ is an instance of \textsc{Independent Set}.
	Let $k=|V(G)|-k'$.
	Transform the instance $(G,k)$ of \textsc{Independent Set} into an equivalent instance $(G', p, k)$, where $G'$ is a graph obtained from $G$ by subdivision of each edge of $G$.
	The precoloring $p$ colors each vertex of $G$ in $G'$ with an unique color corresponding to this vertex, i.e.\ $p(v)=v$ for each $v \in V(G)$.
	The vertices of $G'$ that are introduced because of the subdivision are left uncolored.
	Thus, the number of colors used equals $\ell=|V(G)|$.
	Note that all vertices of $G$ are potentially happy in $(G',p)$, so the number of potentially happy vertices in $(G',p)$ also equals $h=\ell=|V(G)|$.
	The following claim shows that the constructed instance $(G',p,k)$ is equivalent to the instance $(G,k)$ of \textsc{Independent Set}.
	
	\begin{claim}
		For any $S\subseteq V(G)$, $S$ is an independent set in $G$ if and only if all vertices in $S$ can be simultaneously happy in $(G', p)$.
	\end{claim}
	\begin{claimproof}
		Let $S$ be an independent set in $G$.
		We construct a coloring $c$ extending $p$ as follows.
		For each uncolored vertex $e_{uv}$ of $G'$, that corresponds to the subdivision of the edge $uv \in E(G)$, put
		$$c(e_{uv})=\left\{
		\begin{matrix}
		p(u), && \text{if $u \in S$,} \\
		p(v), &&\text{if $v \in S$,} \\
		\text{any}, && \text{otherwise.}
		\end{matrix}
		\right.$$
		
		Since $S$ is an independent set, $u \in S$ and $v \in S$ never hold simultaneously. Thus, for each $v \in S$ and each $e_{uv} \in N_{G'}(v)$, $c(e_{uv})=p(v)=c(v)$. Therefore, all vertices in $S$ are happy with respect to $c$.
		
		In the other direction, let $c$ be a coloring of $G'$ extending $p$.
		Firstly, note that no newly-introduced vertex $e_{uv} \in V(G')\setminus V(G)$ can be happy.
		$e_{uv}$ is adjacent to vertices $u$ and $v$ in $G'$, but $p(u)\neq p(v)$, as $uv$ is an edge of $G$ and $p$ corresponds to a proper coloring of $G$.
		Hence, if $S$ is a subset of vertices that are happy with respect to $c$, then $S \subseteq V(G)$ necessarily.
		Suppose now that $S$ is not an independent set in $G$, i.e.\ $u, v \in S$, but $uv \in E(G)$.
		Consider the vertex $e_{uv}$ in $G'$.
		Since both $u$ and $v$ are happy with respect to $c$, $c(u)=c(e_{uv})$ and $c(v)=c(e_{uv})$.
		But $p(u)\neq p(v)$, a contradiction.
		The proof of the claim is finished.
	\end{claimproof}
	
	Finally, $\mathcal{A}'$ outputs the instance $(G',p,k)$ as an instance of a language parameterized either by $k^d \ell$ or by $k^d h$, i.e.\ $((G',p,k), k^d \ell)$ or $((G',p,k), k^d h)$.
	Since $h=\ell$, these parameters are equal.
	We now show that $\mathcal{A}'$ is a weak cross-composition of dimension $2$.
	We already proved that the instance output by $\mathcal{A}'$ is equivalent to the instance output by $\mathcal{A}$, so the condition \ref{enum:weak_cc_equiv} of weak cross-compositions is satisfied.
	It suffices to prove that the condition \ref{enum:weak_cc_poly_bound} is satisfied as well.
	
	Note that
	\begin{equation}
	\begin{split}
	k^d h = k^d \ell = (|V(G)|-k')^d \cdot |V(G)| \le (q(\max_{i=1}^t|x_i|))^d \cdot p(\max_{i=1}^t|x_i|) \cdot \sqrt{t}\\=(q^d p)(\max_{i=1}^t|x_i|) \cdot \sqrt{t},
	\end{split}
	\end{equation}
	 where $p$ and $q$ are polynomials from Lemma \ref{lemma:book_weak_cross_comp_vc}.
	Thus, $\mathcal{A}'$ satisfies the condition \ref{enum:weak_cc_poly_bound} of weak cross-compositions.
	Hence, $\mathcal{A}'$ is the desired weak cross-composition from \textsc{Multicolored Biclique} to \MHV, and theorem follows.
\end{proofO}

	\section{Parameterization above guarantee}\label{sec:above-guarantee}

This section concerns the above guarantee parameter for \MHV~and \MHE.
By \textit{guarantee} we mean the number of happy vertices or edges that can be obtained with a trivial extension of the precoloring given in input.
The definition of trivial extensions follows.
\begin{definition}
For a graph with precoloring $(G,p)$, we call a full coloring $c$ a \textit{trivial extension} of $p$, if $p$ can be extended to $c$ by choosing a single color $i$ and assigning color $i$ to every uncolored vertex.
In other words, $p(v)=c(v)$ for every $v \in p^{-1}([\ell])$, and $c(u)=c(v)$ for every $u, v \notin p^{-1}([\ell])$.
\end{definition}
We formulate the version of \MHV~where the above guarantee parameter equals one.

\begin{problemx}
	\problemtitle{\scshape Above Guarantee Happy Vertices}
	\probleminput{A graph $G$, a partial coloring $p: S \to [\ell]$ for some $S \subseteq V(G)$ and integer $k$, such that there is a trivial extension of $p$ that yields exactly $k$ happy vertices in $G$.}
	\problemquestion{Is $(G, p, k+1)$ a yes-instance of \MHV?}
\end{problemx}

The \textsc{Above Guarantee Happy Edges} is formulated analogously.
We show that both these problems cannot be solved in polynomial time, unless $\P=\NP$.
We start with \textsc{Above Guarantee Happy Vertices}.
To prove that it is computationally hard, we provide a chain of polynomial reductions.
An intermediate problem in this chain is the \textsc{Weighted MAX-2-SAT} problem.

\begin{problemx}
	\problemtitle{\textsc{Weighted MAX-2-SAT}}
	\probleminput{A boolean formula in $2$-CNF with integer weights assigned to its clauses, an integer $w$.}
	\problemquestion{Is there an assignment of the variables of $\phi$ satisfying clauses of total weight at least $w$ in $\phi$?}
\end{problemx}

\begin{lemmaO}\label{lemma:w_max_2_sat_npc}
	\textsc{Weighted MAX-2-SAT} is \NP-complete even when the inputs $\phi$ and $w$ satisfy
	\begin{enumerate}
		\item The total weight of all positive clauses (i.e., clauses containing at least one positive literal) of $\phi$ equals $w-1$;\label{lemma:cond:w_max_2_sat_pos_clauses}
		\item Each clause of $\phi$ is assigned either weight $1$ or weight $13$;\label{lemma:cond:w_max_2_sat_constant_weights}
		\item Each variable appears exactly three times in $\phi$, at least once positively in a clause containing also a negative literal, and at least once negatively in a clause containing also a positive literal.\label{lemma:cond:w_max_2_sat_three_occs}
	\end{enumerate}
\end{lemmaO}
\begin{proofO}	
	The proof is by a chain of technical polynomial reductions from \textsc{3-SAT}, that is a classical \NP-complete problem. Most of the reductions below are classical, but we carefully follow them to ensure that intermediate formulas have certain important properties.
	
	Let $\phi_0=C_1 \wedge C_2 \wedge \ldots \wedge C_m$ be a formula in CNF on $n$ variables consisting of $m$ clauses, and each clause consists of no more than three literals. Transform the formula to $\phi_1=(C_1 \lor x) \land (C_2 \lor x) \land \ldots \land (C_m \lor x) \land \overline{x}$, where $x$ is a variable new to the formula. We obtain an equivalent \textsc{4-SAT} input formula on $n+1$ variables and $m+1$ clauses, where each clause except one contains a positive literal. We can trivially satisfy $m$ clauses of the formula, but if all $m+1$ clauses of $\phi_1$ can be satisfied, then all clauses of the initial formula $\phi_0$ can be satisfied as well.
	
	Now transform the 4-CNF formula $\phi_1$ to a 3-CNF formula $\phi_2$, so it contains just one negative clause as well as $\phi_1$. Do it as follows. Leave clauses of $\phi_1$, that consist of at most three literals, as is, and introduce them to $\phi_2$. For each clause of $\phi_1$ consisting of four literals, replace it with two clauses, introducing a new variable  specific to this clause. That is, take a clause of $\phi_1$ of length four, say $a \lor b \lor c \lor x$, where $a, b, c$ are some literals (may be positive as well as negative), but $x$ is necessarily a positive literal. Then, take a new variable $z_i$, and introduce clauses $(a \lor b \lor z_i)\land (\overline{z_i} \lor x)$ to $\phi_2$. Note that these two clauses are both positive. Hence, $\phi_2$ is a 3-CNF formula equivalent to $\phi_1$, and the only negative clause in $\phi_2$ is $\overline{x}$.
	
	We then need each variable in $\phi_2$ to appear at most three times. We use a standard technique to achieve that. If a variable $y$ appears $k>3$ times in $\phi_2$, introduce $k$ new variables $y_1, y_2, \ldots, y_k$ to $\phi_2$, replace $i^{\text{th}}$ occurence of $y$ with a variable $y_i$, and introduce $k$ clauses $(\overline{y_1} \lor y_2) \land (\overline{y_2} \lor y_3) \land \ldots \land (\overline{y_k} \lor y_1)$, so that newly-introduced variables are equal in any satisfying assignment. Note that no new negative clause is introduced, so the obtained formula $\phi_3$ is a formula in 3-CNF equivalent to $\phi_2$, with only one negative clause, and each variable appears in $\phi_3$ at most three times. It is possible to satisfy all except one clause of $\phi_3$ simultaneously, but it is \NP-complete to decide whether one can satisfy the entire formula. The only negative clause (that is, a clause consisting only of negative literals) in $\phi_3$ consists of a single literal.
	
	We now transform our special instance of \textsc{3-SAT} into a special instance of \textsc{MAX-2-SAT}.
	In our special case, all clauses of length three are positive.
	We use the classical reduction from \textsc{3-SAT} to \textsc{MAX-2-SAT} \cite{Garey1976}: given an initial formula in 3-CNF, we replace each clause of the formula consisting of exactly three literals, say $a \lor b \lor c$, with ten clauses, introducing a new variable $z_i$: $$z_i \land a \land b \land c \land (\overline{a} \lor \overline{b}) \land (\overline{b}\lor \overline{c})\land (\overline{a} \lor \overline{c})\land (\overline{z_i}\lor a)\land (\overline{z_i}\lor b)\land (\overline{z_i}\lor c).$$
	These ten clauses has a property that, an assignment of the variables of $a, b, c$ satisfies $a \lor b \lor c$  if and only if the same assignment satisfies exactly seven clauses out of these ten clauses, with at least one of the two possible assignments of $z_i$.
	Also, no more than seven clauses can be satisfied simultaneously among these ten clauses.
	
	In our case, $a \lor b \lor c$ is a positive clause.
	If there are three or two positive literals among $a, b, c$, replace $a\lor b \lor c$ with the ten clauses above, introducing a new variable $z_i$.
	If there is only one positive literal among $a, b, c$, replace $a\lor b\lor c$ instead with the following ten clauses, where the literals of $z_i$ are negated: $$\overline{z_i} \land a \land b \land c \land (\overline{a} \lor \overline{b}) \land (\overline{b}\lor \overline{c})\land (\overline{a} \lor \overline{c})\land ({z_i}\lor a)\land ({z_i}\lor b)\land ({z_i}\lor c).$$
	
	Note that after such replacement, the all-true assignment of the variables of $a, b, c$ and $z_i$ satisfies exactly seven out of the ten clauses above.
	
	We do not change any clause consisting of less than three literals. Let the initial formula consist of $m=m_1+m_2+m_3$ clauses, where $m_i$ is the number of clauses consisting of exactly $i$ literals in the initial formula. Then, after the transformation, we obtain $m_1+m_2+10m_3$ clauses consisting of at most two literals, and we ask to satisfy at least $m_1+m_2+7m_3$ of them simultaneously. 
	
	Now continue our chain of reductions and apply the described reduction to $\phi_3$. We obtain a 2-CNF formula $\phi_4$ that consists of $m_1+m_2+10m_3$ clauses of length at most two, where $m_i$ is the number of clauses of length $i$ in $\phi_3$. Consider the all-true assignment in $\phi_3$. It satisfies all clauses of $\phi_3$, except the single clause consisting of one negative literal. Then, the all-true assignment in $\phi_4$ (including the variables that are newly-introduced in $\phi_4$) satisfies exactly $(m_1-1)+m_2+7m_3$ clauses of $\phi_4$. Thus, we again obtain an \NP-hard problem of satisfying one more clause of the formula than in the all-true assignment. And now, our formula is in 2-CNF. Moreover, each variable appears at most three times in $\phi_3$, and it gets copied at most four times in $\phi_4$. Therefore, our formula $\phi_4$ is a 2-CNF formula that also has a property that each variable appears at most twelve times in it.
	
	Then, we again reduce the number of occurences of a variable in our formula. We again do that in the standard way: for each variable $y$ that occurs $k$ times (even for $k<3$), we introduce new variables $y_1, y_2, \ldots, y_k$, replace $i^{\text{th}}$ occurence with $y_i$ for each $i \in [k]$, and introduce new clauses $(\overline{y_1}\lor y_2)\land (\overline{y_2}\lor y_3)\land \ldots \land (\overline{y_k} \lor y_1)$. The obtained formula $\phi_5$ is in 2-CNF and each variable appears in $\phi_5$ at most three times, but we can not just ask to satisfy one more clause in $\phi_5$ than in the all-true assignment. The reason behind this is that now we do not ask to satisfy the whole formula, so the clauses of type $(\overline{y_1}\lor y_2)\land (\overline{y_2}\lor y_3)\land \ldots \land (\overline{y_k} \lor y_1)$ might not be satisfied completely in an appropriate assignment. Thus, it might be the case that in an assignment $\sigma$ that satisfies a sufficient number of clauses of $\phi_5$, $\sigma(y_i) \neq \sigma(y_j)$ for some pair of variables corresponding to the same variable $y$ of the initial formula $\phi_4$. So the obtained problem of satisfying clauses of $\phi_5$ is not equivalent to the initial one.
	
	To overcome this difficulty, we assign weights to the clauses of $\phi_5$, making the newly-introduced clauses of type $(\overline{y_i} \lor y_{i+1})$ weigh more than the regular clauses that come from $\phi_4$. Specifically, we assign weight $13$ to each newly-introduced clause of $\phi_5$, and we assign weight $1$ to each remaining clause of $\phi_5$. $\phi_4$ consists of $m_1 + m_2 + 10m_3$ clauses, and $(m_1-1)+m_2+7m_3$ clauses are satisfied in $\phi_4$ with the all-true assignment. Let $t$ be the number of newly-introduced clauses in $\phi_5$. Then, $\phi_5$ consists of $m_1+m_2+10m_3+t$ clauses, and the total weight of clauses that are satisfied in $\phi_5$ with the all-true assignment is $(m_1-1)+m_2+7m_3+13t$. We ask whether it is possible to satisfy clauses of total weight at least $m_1+m_2+7m_3+13t$ in $\phi_5$. We claim that this problem is equivalent to the problem of satisfying clauses of $\phi_4$.
	
	\begin{claim}
		There is an assignment satisfying at least $s$ clauses of $\phi_4$ if and only if there is an assignment satisfying clauses of total weight at least $s+13t$ in $\phi_5$.
	\end{claim}
	\begin{claimproof}
		The proof in one direction is trivial. Given an appropriate assignment $\sigma$ of the variables of $\phi_4$, it is easy to construct an appropriate assignment $\sigma'$ for $\phi_5$. Just put $\sigma'(y_i)=\sigma(y)$ for $i^\text{th}$ occurence of a variable $y$ in $\phi_4$. $\sigma'$ satisfies all $s$ clauses in $\phi_5$ corresponding to the clauses satisfied by $\sigma$ in $\phi_4$, and satisfies all $t$ newly-introduced clauses of weight $13$ in $\phi_5$.
		
		In the other direction, take an assignment $\sigma'$ that satisfies the  maximum possible number of clauses in $\phi_5$ simultaneously. We argue that $\sigma'$ satisfies all clauses of weight $13$ in $\phi_5$. Suppose it's not true and for some variable $y$ appearing $k$ times in $\phi_4$, at least one clause in $(\overline{y_1}\lor y_2)\land (\overline{y_2}\lor y_3)\land \ldots \land (\overline{y_k} \lor y_1)$ in $\phi_5$ is not satisfied by $\sigma'$. Apart from two of these $k$ clauses of weight $13$, each $y_i$ appears in exactly one clause of weight $1$ in $\phi_5$, that comes from $\phi_4$ initially. Recall that each variable of $\phi_4$ appears at most $12$ times in $\phi_4$. Hence, the variables $y_1, y_2, \ldots, y_k$ together touch at most $12$ clauses of weight $1$ in $\phi_5$. Change $\sigma'$ by setting $\sigma'(y_i)=0$ for each $i \in [k]$. Some clauses of weight $1$ may become unsatisfied, but there are at most $12$ of them. Thus, at most $12$ weight is lost with the change. At the other hand, all clauses of weight $13$ become satisfied, and at least $13$ weight is gained with the change. At least $1$ weight is gained with the change of $\sigma'$ in total --- a contradiction with the optimality of $\sigma'$.
		
		Thus, if $\sigma'$ satisfies clauses of total weight at least $s+13t$ in $\phi_5$, we may be sure that all $t$ clauses of weight $13$ are satisfied by $\sigma'$. That is, $\sigma'(y_i)=\sigma'(y_j)$ holds for each pair $y_i, y_j$ of variables corresponding to occurences of the same variable $y$ in $\phi_4$. We get that an assignment $\sigma$, constructed by $\sigma(y)=\sigma'(y_1)$, satisfies at least $s$ clauses of $\phi_4$. The claim statement follows. 
	\end{claimproof}

	We conclude that obtained problem of satisfying clauses of $\phi_5$ is equivalent to the initial problem of satisfying $\phi_0$.
	Moreover, $\phi_5$ satisfies the lemma conditions.
	All reductions presented are polynomial, and the lemma statement follows.
\end{proofO}

The chain continues with the following version of the \textsc{Independent Set} problem.

\begin{problemx}
	\problemtitle{\textsc{Independent Set Above Coloring}}
	\probleminput{A graph $G$, properly colored with $\ell$ colors: $V(G)=V_1\sqcup V_2 \sqcup \ldots \sqcup V_\ell$.}
	\problemquestion{Is there an independent set of size at least $\max\limits_{i=1}^\ell |V_i|+1$ in $G$?}
\end{problemx}

\begin{lemmaO}\label{lemma:is_above_coloring_npcomp}
	\textsc{Independent Set Above Coloring} is \NP-complete for $\ell=3$.
\end{lemmaO}
\begin{proofO}
	We reduce from the special case of \textsc{Weighted MAX-2-SAT} from the statement of Lemma \ref{lemma:w_max_2_sat_npc}.
	We reduce to the weighted version of \textsc{Independent Set} first, and then show how to get rid of the weights.
	
	Let be given a formula $\phi$ and an integer $w$ as an instance of \textsc{Weighted MAX-2-SAT} satisfying the conditions of Lemma \ref{lemma:w_max_2_sat_npc}.
	We construct a graph $G$ with weights assigned to its vertices in the same way as in the classical reduction from \textsc{Satisfiability} to \textsc{Clique} (as a complement of \textsc{Independent Set}) by Cook \cite{Cook1971} or Karp \cite{Karp1972}.
	That is, for each literal in $\phi$, we introduce a new vertex in $G$.
	Since the clauses in $\phi$ are weighted, we assign each vertex a weight equal to the weight of the clause of the corresponding literal. Then, for each clause of length two in $\phi$, we connect the vertices corresponding to its literals by an edge in $G$.
	Finally, we connect each pair of vertices in $G$ that correspond to opposite literals of the same variable in $\phi$ by an edge.
	A claim follows.
	
	\begin{claim}\label{claim:sat_is_reduction}
		$\phi$ has an assignment satisfying clauses of total weight at least $w$ if and only if $G$ has an independent set of total weight at least $w$.
	\end{claim}
	
	We now show how to color $G$ with three colors properly.
	Firstly, for each positive clause of $\phi$, take an arbitrary positive literal of this clause and color the corresponding vertex in $G$ with color $1$.
	Thus the set $V_1$ of the vertices colored with color $1$ in $G$ becomes constructed.
	Note that by condition \ref{lemma:cond:w_max_2_sat_pos_clauses} of Lemma \ref{lemma:w_max_2_sat_npc}, the total weight of the vertices in $V_1$ equals $w-1$.
	Moreover, $V_1$ forms an independent set in $G$, since its vertices correspond to positive literals from pairwise different clauses.
	
	Then, take each positive clause of $\phi$ that contains a negative literal.
	That is, a clause consisting of one negative and one positive literal.
	Color the negative literal of the clause with color $2$.
	Note that for now vertices of color $2$ form an independent set in $G$, since they correspond to negative literals from pairwise different clauses.
	
	By condition \ref{lemma:cond:w_max_2_sat_three_occs} of Lemma \ref{lemma:w_max_2_sat_npc}, each variable of $\phi$ has three corresponding vertices in $G$.
	Moreover, at least two corresponding vertices are already colored in $G$: at least one corresponding to a positive occurence in a negative clause is colored with color $1$ and at least one corresponding to a negative occurence in a positive clause is colored with color $2$.
	Hence, for each variable of $\phi$, at most one vertex corresponding to this variable is uncolored.
	
	Take any uncolored vertex $v$ in $G$.
	It corresponds either to a positive literal, say $x$, or a negative literal, say $\overline{x}$.
	If it corresponds to a positive literal $x$, then this literal is a literal from a clause consisting of two positive literals, otherwise $v$ would be colored with color $1$.
	Let this clause be $x \lor y$.
	Since the clause is positive and the vertex $v$ corresponding to $x$ is uncolored, then the vertex corresponding to $y$ is colored with color $1$.
	This vertex is connected to $v$ by an edge, so color $v$ with color $3$.
	
	If $v$ corresponds to a literal $\overline{x}$, then this literal is necessarily from a negative clause.
	If this clause consists just of $\overline{x}$, color $v$ with color $2$ or color $3$ arbitrarily.
	Otherwise, the clause is of length two, say $\overline{x} \lor \overline{y}$.
	The vertex corresponding to $\overline{y}$ is not colored, since it is not from a positive clause.
	Both vertices corresponding to $\overline{x}$ and $\overline{y}$ are uncolored, so color one of them with color $2$ and the other with color $3$ arbitrarily.
	Note that the vertices colored with color $2$ still correspond to negative literals from pairwise different clauses, hence they still form an independent set in $G$.
	The vertices of color $3$ correspond to pairwise different variables from pairwise different clauses, so they form an independent set in $G$ as well.
	
	We colored $G$ with three colors, that is, partitioned $V(G)$ into $V(G)=V_1 \sqcup V_2 \sqcup V_3$, where each of $V_1, V_2, V_3$ forms an independent set in $G$.
	Though the constructed graph $G$ has weights assigned to its vertices.
	Note that positive integer clauses weights can be avoided in \textsc{Weighted MAX-2-SAT} just by replacing clause each clause of weight $t$ with $t$ copies of this clause in the formula.
	By condition \ref{lemma:cond:w_max_2_sat_constant_weights} of Lemma \ref{lemma:w_max_2_sat_npc}, some clauses get copied $13$ times, and the others remain appearing just once in the formula.
	Then, the reduction from \textsc{MAX-2-SAT} to \textsc{Independent Set} is the same as for the weighted versions of the problems.
	
	The thing why we needed weights is to simiplify the coloring of $G$ with three colors.
	It is easy to see that, after replacing weights with copies in $\phi$, $G$ remains the same, just vertex weights become replaced with vertex copies.
	More importantly, no edge is added between different copies of the same vertex. Thus, the new, unweighted graph $G'$ obtained by a reduction from the new, unweighted formula $\phi'$, can be colored with three colors in the same way as $G$ can.
	It's just that all vertices in $G'$ that are copies of the same vertex in $G$ receive the same color as their original vertex in $G$.
	Vertices of $G'$ becomes partitioned in three independent sets $V(G')=V'_1\sqcup V'_2 \sqcup V'_3$, and $|V'_1|=w-1$. Moreover, by Lemma \ref{lemma:w_max_2_sat_npc} and Claim \ref{claim:sat_is_reduction}, finding an independent set of size at least $w$ in $G'$ is an \NP-complete problem. All reductions and algorithms provided in the proof are polynomial, and the lemma follows.
\end{proofO}

\begin{theoremO}\label{thm:mhv_above_guar}
	\textsc{Above Guarantee Happy Vertices} is \NP-complete even when $\ell=3$.
\end{theoremO}
\begin{proofO}
	We reduce from \textsc{Independent Set Above Coloring} with $\ell=3$, that is \NP-complete by Lemma \ref{lemma:is_above_coloring_npcomp}.
	Let $G$ be an instance of \textsc{Independent Set Above Coloring}.
	That is, $G$ is colored properly with three colors: $V(G)=V_1 \sqcup V_2 \sqcup V_3$; and it is asked to find an independent set of size at least $\max\limits_{i=1}^3|V_i|+1$ in $G$.
	
	We construct an instance $(G',p,k)$ of \textsc{Above Guarantee Happy Vertices} as follows.
	Obtain $G'$ as a subdivision of $G$.
	Construct the partial coloring $p$ as follows.
	Left all new vertices appeared after subdivision uncolored.
	For each other vertex, that is, for each vertex in $V(G)$, precolor it with the same color as it is colored in $G$.
	That is, for each $i \in [3]$, and for each $v \in V_i$, put $p(v)=i$. Finally, put $k=\max\limits_{i=1}^3 |V_i|$.
	Note that the reduction is done in polynomial time.
	We formulate that the constructed instance of \textsc{Above Guarantee Happy Vertices} is equivalent to the initial instance of \textsc{Independent Set Above Coloring} in the following claim.
	
	\begin{claim}
		For any $S\subseteq V(G)$, $S$ is an independent set in $G$ if and only if all vertices in $S$ can be simultaneously happy in $(G', p)$.
	\end{claim}
	\begin{claimproof}
		Let $S$ be an independent set in $G$.
		We construct a coloring $c$ extending $p$ as follows.
		For each uncolored vertex $e_{uv}$ of $G'$, that corresponds to the subdivision of the edge $uv \in E(G)$, put
		$$c(e_{uv})=\left\{
		\begin{matrix}
		p(u), && \text{if $u \in S$,} \\
		p(v), &&\text{if $v \in S$,} \\
		\text{any}, && \text{otherwise.}
		\end{matrix}
		\right.$$
		
		Since $S$ is an independent set, $u \in S$ and $v \in S$ never hold simultaneously. Thus, for each $v \in S$ and each $e_{uv} \in N_{G'}(v)$, $c(e_{uv})=p(v)=c(v)$. Therefore, all vertices in $S$ are happy with respect to $c$.
		
		In the other direction, let $c$ be a coloring of $G'$ extending $p$.
		Firstly, note that no newly-introduced vertex $e_{uv} \in V(G')\setminus V(G)$ can be happy.
		$e_{uv}$ is adjacent to vertices $u$ and $v$ in $G'$, but $p(u)\neq p(v)$, as $uv$ is an edge of $G$ and $p$ corresponds to a proper coloring of $G$.
		Hence, if $S$ is a subset of vertices that are happy with respect to $c$, then $S \subseteq V(G)$ necessarily.
		Suppose now that $S$ is not an independent set in $G$, i.e.\ $u, v \in S$, but $uv \in E(G)$.
		Consider the vertex $e_{uv}$ in $G'$.
		Since both $u$ and $v$ are happy with respect to $c$, $c(u)=c(e_{uv})$ and $c(v)=c(e_{uv})$.
		But $p(u)\neq p(v)$, a contradiction.
		The proof of the claim is finished.
	\end{claimproof}

	It is then easy to see that with a trivial extension of $p$ with color $i$ one can obtain exactly $|V_i|$ happy vertices in $G'$.
	Hence, $k$ is indeed a number of happy vertices that can be obtained with a trivial extension of $p$.
	Finally, finding a coloring yielding at least $k+1$ happy vertices in $(G', p)$ is equivalent to finding an independent set of size $\max_{i=1}^3 |V_i|+1$.
	Thus, \textsc{Above Guarantee Happy Vertices} is \NP-complete for $\ell=3$.
\end{proofO}

We now turn onto \textsc{Above Guarantee Happy Edges}.
We provide a reduction from the following well-known \NP-complete problem.

\begin{problemx}
	\problemtitle{\textsc{Exact 3-Cover (X3C)} \cite{Garey1976, Karp1972}}
	\probleminput{An integer $n$, a collection $\mathcal{S}=\{S_1, S_2, \ldots, S_m\}$ of three-element subsets of $[3n]$.
	\problemquestion{Is there an exact cover of $[3n]$ with elements of $\mathcal{S}$, i.e.\ is there a sequence $i_1, i_2, \ldots, i_n$, such that $S_{i_1}\cup S_{i_2}\cup \ldots \cup S_{i_n}=[3n]$?}}
\end{problemx}

\begin{theoremO}\label{thm:mhe_above_guar}
	\textsc{Above Guarantee Happy Edges} is \NP-complete.
\end{theoremO}
\begin{proofO}
	We reduce from the \textsc{Exact 3-Cover} problem. Let $(n, \mathcal{S}=(S_1, S_2,$ $ \ldots, S_m))$ be an instance of \textsc{X3C}. In our reduction, we need $n$ to be an odd number. If $n$ is even, we can always increase $n$ by one and add the set $\{3n+1, 3n+2, 3n+3\}$ to $\mathcal{S}$ and obtain an equivalent instance of \textsc{X3C}. We also assume that each element of $[3n]$ is contained in at least one set in $\mathcal{S}$. Now we construct an instance $(G, p, k)$ of \MHE~as follows.
	
	For each integer $i \in [3n]$, introduce a new uncolored vertex $u_i$ to $G$. Then, for each $i, j \in [3n]$, $i\neq j$, introduce an edge between $u_i$ and $u_j$ in $G$, so $u_1, u_2, \ldots, u_{3n}$ form a clique in $G$. These vertices correspond to the elements in $[3n]$.
	
	For each set $S_i \in \mathcal{S}$, introduce $3n$ new vertices in $G$, namely $v_{i,1}, v_{i,2}, \ldots, v_{i, 3n}$. Each of these vertices we precolor with color $i$, i.e.\ $p(v_{i,j})=i$ for each $j \in [3n]$. We also connect each of these vertices to all vertices corresponding to the elements of $S_i$, i.e.\ introduce an edge between $v_{i,j}$ and $u_t$ for each $j \in [3n]$ and each $t \in S_i$.
	
	Finally, we introduce a group of new $3(n+1)/2$ vertices $w_1, w_2, \ldots, w_{3(n+1)/2}$ in $G$. Each of them we precolor with color $m+1$. We also introduce an edge $(w_i, u_j)$ to $G$ for every $i \in [3(n+1)/2]$ and every $j \in [3n]$, except for the edge $(w_1, u_1)$. Hence, we introduce $3(n+1)/2 \cdot 3n - 1$ such edges.
	
	We then set $k=9n^2+3n-1$ and say that this number of happy edges in $G$ can be obtained by coloring every uncolored vertex with color $m+1$. Indeed, say that $u_j$ is colored with color $m+1$ for every $j\in[3n]$. These vertices form clique in $G$, hence all $\binom{3n}{2}$ edges of the clique are happy. Only edges left that are happy are edges of type $(w_i, u_j)$. Recall that we introduced $3(n+1)/2 \cdot 3n - 1$ such edges, hence we get that
	$$\binom{3n}{2}+3(n+1)/2\cdot 3n-1=3n\cdot\left((3n-1)/2+3(n+1)/2\right)-1=9n^2+3n-1$$
	edges are happy in $G$ with respect to such trivial extension of $p$.
	
	We now argue that $(n, \mathcal{S})$ is a yes-instance of \textsc{X3C} if and only if $(G, p, k+1)$ is a yes-instance of \MHE. Note that we use $m+1$ colors, i.e.\ $\ell=m+1$.
	
	Let prove that if $(n, \mathcal{S})$ is a yes-instance of \textsc{X3C}, then $(G, p, k+1)$ is a yes-instance of \MHE. Let $i_1, i_2, \ldots, i_n$ be an answer to $(n, \mathcal{S})$. Then for each $j \in [3n]$, there is a unique $t(j) \in [n]$, such that $t(j) \in \{i_1, i_2, \ldots, i_n\}$ and $j \in S_{t(j)}$. Extend $p$ with a coloring $c$ such that $c(u_j)=t(j)$ for each $j \in [3n]$. We claim that there are exactly $k+1$ happy edges in $G$ with respect to $c$.
	
	$G$ consists of edges that have exactly one endpoint of type $u_j$ and edges of type $(u_{j_1}, u_{j_2})$ for $j_1\neq j_2$. Each $u_j$ is connected to exactly $3n$ vertices $v_{t(j), 1}, v_{t(j), 2}, \ldots, v_{t(j), 3n}$ of color $t(j)$, since $u_j \in S_{t(j)}$. Hence there are exactly $9n^2$ happy edges in $G$ that have exactly one endpoint of type $u_j$, with respect to $c$. Observe that an edge $(u_{j_1}, u_{j_2})$ is happy with respect to $c$ if and only if $t(j_1)=t(j_2)$, i.e.\ $j_1$ and $j_2$ are covered by the same set $S_{t(j_1)}$. Since the size of every set in $\mathcal{S}$ is exactly three, each $u_j$ is adjacent to exactly two vertices of type $u_{j'}$ of the same color. Thus, there are exactly $3n$ happy edges of type $(u_{j_1}, u_{j_2})$ in total in $G$, with respect to $c$. Happy edges of both types sum up to a total of $9n^2+3n$ happy edges.
	
	Let now prove in the other direction. Let $c$ be an optimal coloring of $G$ extending $p$ such that at least $k+1$ edges are happy in $G$ with respect to $c$.
	
	\begin{claim}\label{claim:he_triv_req_cols}
		For any optimal coloring $c$ of $G$ and any $j \in [3n]$, either $c(u_j)=m+1$ or $u_j \in S_{c(u_j)}$.
	\end{claim}
	\begin{claimproof}
		Suppose it's not true, and $c$ is an optimal coloring of $G$ and $c(u_j) \neq m+1$ and $u_j \notin S_{c(u_j)}$ for some $j \in [3n]$. Then only happy edges that are incident to $u_j$ are edges of the clique, since any other edge incident to $u_j$ has the other endpoint precolored with either color $m+1$ or color $i$ for any set $S_i$ containing $j$. Hence, $u_j$ is incident to at most $3n-1$ happy edges with respect to $c$.
		
		Now change the color of $u_j$ to any color $i$ such that $j \in S_i$. Such color exists, since we assumed that any element in $[3n]$ is contained in at least one set. $u_j$ is adjacent to $3n$ vertices of color $i$, hence $u_j$ now is incident to at least $3n$ happy edges. Thus, such change in $c(u_j)$ allows to win at least one happy edge. A contradiction with optimality of $c$.
	\end{claimproof}

	Let $s$ be the number of vertices among $u_1, u_2, \ldots, u_{3n}$ colored with color $m+1$ in $c$, i.e.\ $s=|c^{-1}(m+1)\cap \{u_1, u_2 ,\ldots, u_{3n}\}|$. The only other vertices colored with color $m+1$ are vertices $w_1, w_2, \ldots, w_{3(n+1)/2}$, hence there are at most $\binom{s}{2}+s\cdot 3(n+1)/2$ happy edges incident to vertices of color $m+1$.
	
	By Claim \ref{claim:he_triv_req_cols}, every other vertex of type $u_j$ (there are $3n-s$ of them) is colored with a color corresponding to a set containing $j$. Thus, each such vertex is adjacent to $3n$ precolored vertices of the same color and to at most two vertices of the same color in the clique. Hence, there are at most $(3n-s) \cdot 3n + (3n-s)$ happy edges not incident to vertices of color $m+1$ in $G$ with respect to $c$.
	
	In total, we get that at most
	$$\binom{s}{2}+s \cdot 3(n+1)/2 +(3n-s)\cdot (3n+1)=9n^2+3n+\binom{s}{2}+s \cdot 3(n+1)/2-s\cdot(3n+1)$$
	edges are happy in $G$ with respect to $c$. Recall that $c$ is a coloring yielding at least $9n^2+3n$ happy edges in $G$, hence
	$$\binom{s}{2}+s \cdot 3(n+1)/2-s\cdot(3n+1)\ge 0,$$
	$$s \cdot \left((s-1)/2+3(n+1)/2-(3n+1)\right) \ge 0,$$
	$$s \cdot \left((s-1)+3(n+1)-(6n+2)\right)\ge 0,$$
	$$s \cdot (s - 3n) \ge 0.$$
	
	The inequality above holds only when $s \le 0$ or $s \ge 3n$, but since $s \in \{0, 1, \ldots, 3n\}$, it is either $s=0$ or $s=3n$. From the construction of $(G, p)$ we already know that when $s=3n$, only $k=9n^2+3n-1$ edges are happy in $G$, hence the only option left is $s=0$. Thus, each vertex of type $u_j$ is colored with a color corresponding to a set containing $j$, but not with color $m+1$.
	
	As we observed earlier, each $u_j$ is adjacent to $3n$ precolored vertices of the same color and to at most two vertices of the same color among vertices of the clique. Hence, the only way to obtain $9n^2+3n$ happy edges is when each $u_j$ has \textit{exactly} two neighbours of the same color in the clique. Thus, for each color presented among $c(u_j)$, there are exactly three vertices of such color in the clique. This yields a solution of the initial instance of \textsc{X3C}: it is sufficient to take sets with indices in the set $\{c(u_1), c(u_2), \ldots, c(u_{3n})\}$. This finishes the proof of the theorem.
\end{proofO}

	\section{ETH and Set Cover Conjecture based lower bounds}

In this section, we show lower bounds for exact algorithms for \MHV~and \MHE, based on the popular Exponential Time Hypothesis and the Set Cover Conjecture.
We start with the Set Cover Conjecture and the following problem.

\begin{problemx}
	\problemtitle{\textsc{Set Partitioning}}
	\probleminput{An integer $n$, a set family $\mathcal{F}=\{S_1, S_2, \ldots, S_m \}$ over a universe $U$ with $|U|=n$.}
	\problemquestion{Is there a sequence of pairwise disjoint sets $S_{i_1}, S_{i_2}, \ldots, S_{i_k}$ in $ \mathcal{F}$, such that $\bigsqcup\limits_{j=1}^{k} S_{i_j} = U$?}
\end{problemx}

\begin{theorem}[\cite{Cygan2016}]
	For any $\epsilon > 0$, \textsc{Set Partitioning} cannot be solved in time $\Ostar{(2-\epsilon)^n}$, unless the Set Cover Conjecture fails.
\end{theorem}
\begin{theoremO}
	For any $\varepsilon > 0$, \MHVfull~cannot be solved in time $\Ostar{(2-\varepsilon)^{n'}}$, where $n'$ is the number of uncolored vertices, unless the Set Cover Conjecture fails.
\end{theoremO}
\begin{proofO}
	The proof is by polynomial reduction preserving the size of the universe of the input instance of \textsc{Set Partitioning} in the number of uncolored vertices of the resulting instance of \MHVfull.
	
	Observe that \textsc{Set Partitioning} is a special case of the weighted version of the \textsc{Set Packing} problem: for each $j \in [m]$, assign weight $|S_j|$ to the set $S_j$ and ask to find a sequence of disjoint sets of summary weight at least $n$. With this observation, adjust the reduction in the proof of Theorem \ref{thm:mhv_nopolycomp} for this special weighted version of \textsc{Set Packing} by introducing $|S_j|$ copies of vertex $s_j$ to $G$ instead of just one.
	This yields the required reduction from \textsc{Set Partitioning} to \MHV.
\end{proofO}

\begin{theoremO}\label{thm:mhe_scc_lb}
	For any $\varepsilon > 0$, \MHEfull~cannot be solved in time $\Ostar{(2-\varepsilon)^{n'}}$, where $n'$ is the number of uncolored vertices, unless the Set Cover Conjecture fails.
\end{theoremO}
\begin{proofO}
	To prove this theorem we would also like to slightly adjust the reduction used to prove the lack of polynomial kernels (Theorem \ref{thm:mhe_nopolycomp}), but there we exploit the restricted version of the problem. The reduction appears possible though, but it is quite more sophisticated than the reduction in the proof of Theorem \ref{thm:mhe_nopolycomp}. We now describe this reduction from \textsc{Set Partitioning} preserving the size of the universe in the number of uncolored vertices in the resulting instance of \MHE. Throughout the proof, we refer to Theorem \ref{thm:mhe_nopolycomp} by saying that we do as usual.
	
	Firstly, we need to get rid of sets consisting of exactly two elements in the initial instance $(U=[n], \mathcal{F}=\{S_1, S_2, \ldots, S_m\})$ of \textsc{Set Partitioning}. We assume that $n>2$. Let $\mathcal{F}_{=2}$ be the subfamily of $\mathcal{F}$ of sets consisting of exactly two elements, i.e.\ $\mathcal{F}_{=2}=\{S_i \mid S_i \in \mathcal{F}, |S_i|=2\}$. We want that $\mathcal{F}_{=2}=\emptyset$. 
	
	Let it be the other case, $\mathcal{F}_{=2}\neq \emptyset$. Obtain an equivalent instance $(U, \mathcal{F}')$ with $\mathcal{F}'_{=2}=\emptyset$ as follows. Start with removing all sets in $\mathcal{F}_{=2}$ from $\mathcal{F}$. Then, for every pair of sets $S_i \in \mathcal{F}$ and $S_j \in \mathcal{F}_{=2}$, such that $S_i \cap S_j = \emptyset$, add $S_i \sqcup S_j$ back in $\mathcal{F}$. That is, return each set of $\mathcal{F}_{=2}$ back in $\mathcal{F}$, but as a union with some disjoint set in $\mathcal{F}$, for each such possible set. Also, for each triple of pairwise disjoint sets in $\mathcal{F}_{=2}$, add their union in $\mathcal{F}$. Formally, 
	\begin{equation}
	\begin{split}
		\mathcal{F}'=\mathcal{F}\setminus\mathcal{F}_{=2}  \cup  \{S_i \sqcup S_j \mid S_i \in \mathcal{F}, S_j \in \mathcal{F}_{=2}, S_i\cap S_j=\emptyset\} \\\cup\{S_i \sqcup S_j \sqcup S_k \mid S_i, S_j, S_k \in \mathcal{F}_{=2}\}.
	\end{split}
	\end{equation}
	Note that $\mathcal{F}'$ does not contain any set of size two and is constructed in polynomial time.
	
	\begin{claim}
		$(U,\mathcal{F})$ and $(U, \mathcal{F}')$ are equivalent instances of \textsc{Set Partitioning}.
	\end{claim}
	\begin{claimproof}
		Let $S_{i_1}, S_{i_2}, \ldots, S_{i_k}$ be the answer to $(U, \mathcal{F})$. If the sequence does not contain sets of size two, then all these sets are contained in $\mathcal{F}'$, hence it is also an answer to $(U, \mathcal{F}')$. Otherwise, there is at least one set of size two in the sequence. Let $S_{i_1}, \ldots, S_{i_{k'}}$ be all sets of size two in the sequence, and $S_{i_{k'+1}}, \ldots, S_{i_k}$ be all other sets in the sequence, $k' \ge 1$. If $k'$ is even, then $S_{i_1}\sqcup S_{i_2}, \ldots, S_{i_{k'-1}} \sqcup S_{i_{k'}}, S_{i_{k'+1}}, \ldots, S_{i_k}$ is an answer to $(U, \mathcal{F}')$. If $k'=1$, then $k\ge 2$, and $S_{i_1} \sqcup S_{i_2}, S_{i_3}, \ldots, S_{i_k}$ is the answer. Otherwise, $k'$ is odd and $k'\ge3$, and $S_{i_1}\sqcup S_{i_2}\sqcup S_{i_3},S_{i_4}\sqcup S_{i_5},\ldots, S_{i_{k'-1}}\sqcup S_{i_{k'}}, S_{i_{k'+1}}, \ldots, S_{i_k}$ is an answer to $(U, \mathcal{F}')$. That is, if $(U, \mathcal{F})$ is a yes-instance, then $(U, \mathcal{F}')$ is a yes-instance.
		
		Proof in the other direction is trivial, since each set of $\mathcal{F}'$ is a disjoint union of a number of sets in $\mathcal{F}$. Thus, the instances are equivalent.
	\end{claimproof}

	We now assume that $\mathcal{F}$ no longer contains any set of size two. We show how to reduce the instance $(U, \mathcal{F})$ of \textsc{Set Partitioning} to an equivalent instance $(G, p, k')$ of \MHE. As usual, introduce a clique on $n$ vertices $u_1, u_2, \ldots, u_n$ in $G$, which vertices correspond to the elements of $U$. For each set $S_j \in \mathcal{F}$, do the following. Let $d=|S_j|$. Introduce $n^2-\lfloor \frac{d-1}{2} \rfloor$ copies of a vertex $s_j$ to $G$. If $d$ is odd, for each $i \in S_j$, add edges between $u_i$ and all $n^2-\frac{d-1}{2}$ copies of $s_j$. If $d$ is even, then divide $S_j$ into two equal parts arbitrarily, say, $S_j=S^1_j\sqcup S^2_j$, $|S^1_j|=|S^2_j|=\frac{d}{2}$. For each $i \in S^1_j$, add edges between $u_i$ and all $n^2-\frac{d}{2}+1$ copies of $s_j$. But for each $i \in S^2_j$, add edges between $u_i$ and all except one copy of $s_j$, that is, connect $u_i$ to only $n^2-\frac{d}{2}$ copies of $s_j$.
	
	As usual, precolor all copies of $s_j$ with color $j$ for each $j \in [m]$, and leave every vertex $u_i$ uncolored. Finally, set $k'=n^3$. We argue that the constructed instance $(G, p, k')$ is a yes-instance if and only if $(U, \mathcal{F})$ is a yes-instance.
	
	To prove in one direction, let $(U, \mathcal{F})$ be a yes-instance and $S_{i_1}, S_{i_2}, \ldots, S_{i_k}$ be the answer sequence. Construct coloring $c$ extending $p$ as usual, by setting $c(u_i)$ equal to the index of the unique set of the answer that contains $i$, i.e.\ $i \in S_{c(u_i)}$, for each $i \in [n]$. Observe that $c$ yields exactly $n^3$ happy edges in $G$. $c$ splits the clique of $G$ into $k$ groups $V_{i_1}, V_{i_2}, \ldots, V_{i_k}$, and each group contains vertices corresponding to the elements of the respective set, i.e.\ $V_{i_t} = \{u_i \mid i\in S_{i_t}\}$ for each $t \in [k]$. Take any group $V_{i_t}$, and let $d=|V_{i_t}|=|S_{i_t}|$. If $d$ is odd, then exactly $d \cdot (n^2+\frac{d-1}{2})=dn^2+\binom{d}{2}$ edges that have exactly one endpoint in $V_{i_t}$ are happy. If $d$ is even, the number of such happy edges equals $\frac{d}{2}\cdot (n^2-\frac{d}{2}+1)+\frac{d}{2}\cdot(n^2-\frac{d}{2})=dn^2-\binom{d}{2}$ as well. The only other edges are $\binom{d}{2}$ edges inside $V_{i_t}$. All vertices in $V_{i_t}$ are of color $i_t$, hence all these $\binom{d}{2}$ edges are happy. Thus, vertices in $V_{i_t}$ are incident to exactly $|V_{i_t}|\cdot n^2$ happy edges. Since no two groups can share endpoints of the same happy edge, we get a total of $\sum |V_{i_t}|\cdot n^2=n^3=k'$ happy edges in $G$ with respect to $c$. Hence, $(G, p, k')$ is a yes-instance of \MHE.
	
	We start the proof in the other direction with a claim identical to Claim \ref{claim:mhe_polycomp_req_cols}.
	
	\begin{claim}\label{claim:mhe_seth_req_cols}
		In any optimal coloring $c$ of $G$ extending $p$, $i \in S_{c(u_i)}$ for each $i \in [n]$.
	\end{claim}
	\begin{claimproof}
		Note that $n^2-\lfloor\frac{n-1}{2}\rfloor>n-1$ for $n>2$, then as usual.
	\end{claimproof}

	The following claim bounds the number of happy edges incident to vertices in each color group of the clique in $G$.
	
	\begin{claim}\label{claim:mhe_seth_main_claim}
		Take any optimal coloring $c$ of $(G, p)$. Let $V_t=c^{-1}(t)\cap \{u_1, u_2, \ldots,$ $ u_n\}$ be the set of vertices of the universe clique of $G$ that are colored with color $t$. Let $d=|V_t|$. Then vertices in $V_t$ correspond to elements in $S_t$, are incident to at most $dn^2$ happy edges, and the bound of $dn^2$ is reached if and only if $d=|S_t|$ or $d=0$.
	\end{claim}
	\begin{claimproof}
		Take any optimal coloring $c$ of $(G, p)$ and $t \in [\ell]$. Let $d=|V_t|$. If $d=0$, then the claim statement holds true. Assume now that $d>0$ and $|V_t|$ is not empty.
		
		By Claim \ref{claim:mhe_seth_req_cols}, all vertices in $V_t$ correspond to elements in $S_t$. Let $D=|S_t|$. Suppose $D$ is odd. Then each vertex in $V_t$ is connected to $n^2-\frac{D-1}{2}$ copies of $s_t$. Hence, exactly $dn^2-d\cdot(D-1)/2$ edges going outside of $V_t$ are happy with respect to $c$. All edges inside $V_t$ are happy, so $V_t$ is incident to $dn^2-d\cdot(D-1)/2+\binom{d}{2}$ happy edges. Since $d \leq D$, this number is not greater than $dn^2$ and equals $dn^2$ if and only if $d=D=|S_t|$. Thus, the claim statement is true for odd $D$.
		
		Otherwise, $D$ is even and $D \ge 4$, since we removed all sets of size two from the initial instance. Recall that for sets of even size, we split them into two halves and connected them to a different number of copies of $s_t$. Let $d=d_1+d_2$ ($d_1, d_2\le \frac{D}{2}$), so $d_1$ vertices in $V_t$ are connected to $n^2-\frac{D}{2}+1$ copies of $s_t$, and the other $d_2$ vertices in $V_t$ are connected to $n^2-\frac{D}{2}$ copies of $s_t$. Then, $V_t$ is incident to exactly $n^2-d\cdot D/2+d_1$ happy edges going outside of $V_t$. Only happy edges left are the $\binom{d}{2}$ edges inside $V_t$.
		
		Suppose that vertices in $V_t$ are incident to at least $dn^2$ happy edges in total, so $d\cdot D/2 - d_1 \le \binom{d}{2}$. Equivalently, $d\cdot (D-(d-1)) \le 2d_1$. Since $2d_1 \le D$, we get that $d\cdot (D-d+1) \le D$, or $d^2-Dd+(D-d) \ge 0$. The quadratic polynomial has roots in $d=1$ and $d=D$, and the inequality holds when $d\le 1$ or $d\ge D$. If $d\le1$, then $d=1$. But then $d \cdot D/2 - d_1 =D/2 - d_1 \ge 1 > \binom{d}{2}=0$. Thus, $d=D=|S_t|$.
		
		The proof is finished. Note that if we would have $D=2$, $V_t$ could consist of one vertex connected to $n^2-\frac{D}{2}+1=n^2$ copies of $s_t$, and the claim statement would fail.
	\end{claimproof}

	Let now $(G, p, k')$ be a yes-instance, and $c$ be an optimal coloring of $(G,p)$. At least $n^3$ edges are happy with respect to $c$ in $G$. By Claim \ref{claim:mhe_seth_main_claim}, it follows that exactly $n^3$ edges are happy in $G$ with respect to $c$, and each color group $V_t$ provides exactly $|V_t|\cdot n^2$ happy edges. Hence, $|V_t|=|S_t|$ for each non-empty group $V_t$. Thus, the color partition of the vertices of the universe clique of $G$ corresponds to a partition of the universe $U$ into sets in $\mathcal{F}$. So $(U, \mathcal{F})$ is a yes-instance of \textsc{Set Partitioning}. 
	We obtained the desired reduction from \textsc{Set Partitioning} to \MHE.
	This finishes the whole proof.
\end{proofO}

We now turn onto ETH-based lower bounds.

\begin{theoremO}
	\MHVfull~with $\ell=3$ cannot be solved in time $2^{o(n+m)}$, unless ETH fails.
\end{theoremO}

\begin{proofO}
	We reuse reductions discussed above in the proofs in Section \ref{sec:above-guarantee}.
	In the proofs of Lemma \ref{lemma:w_max_2_sat_npc} and Lemma \ref{lemma:is_above_coloring_npcomp}, we obtained a chain of linear reductions from \textsc{3-SAT} to \textsc{Independent Set Above Coloring}.
	Then, in the proof of Theorem \ref{thm:mhv_above_guar} we showed that \textsc{Independent Set Above Coloring} can be reduced linearly to \textsc{Above Guarantee Happy Vertices}.
	Since this problem is a special case of \MHV, it follows that \MHV~cannot be solved in $2^{o(n+m)}$ time under ETH.
\end{proofO}

We now prove another computational lower bound for \MHV~that is based on the reduction from \textsc{Independent Set} to \MHV~discussed above in the proofs of Theorem \ref{thm:mhv_kernel_bitsize} and Theorem \ref{thm:mhv_above_guar}.
This reduction also implies some approximation lower bounds.

\begin{theoremO}\label{thm:mhv_clique_eth}
	\MHVfull~cannot be solved in $\O(n^{o(k)})$ time, unless ETH fails.
	Also, for any $\epsilon>0$, \MHVfull~cannot be approximated within $\O(n^{\frac{1}{2}-\epsilon})$, $\O(m^{\frac{1}{2}-\epsilon})$, $\O(h^{1-\epsilon})$ or $\O(\ell^{1-\epsilon})$ in polynomial time, unless $\P=\NP$.
\end{theoremO}
\begin{proofO}
	It is a well-known result that, assuming ETH, both \textsc{Clique} and \textsc{Independent Set} cannot be solved in $n^{o(k)}$ time \cite{Chen2005,Chen2006,lokshtanov2013lower}.
	As discussed in the proofs of Theorem \ref{thm:mhv_kernel_bitsize} and Theorem \ref{thm:mhv_above_guar}, there is a polynomial reduction from \textsc{Independent Set} to \MHV.
	
	Given an instance $(G,k)$ of \textsc{Independent Set}, it is enough to precolor each vertex of $G$ with an unique color and to subdivide each edge of $G$, thus obtaining an equivalent instance $(G',p,k)$ of \MHV.
	$G'$ is a subdivision of $G$, so $|V(G')|=\O(n^2)$, where $n=|V(G)|$.
	Thus, an algorithm with running time $n^{o(k)}$ for \MHV~would imply an algorithm with running time $n^{2o(k)}=n^{o(k)}$ for \textsc{Independent Set}.
	This proves that \MHV~cannot be solved in $n^{o(k)}$ time under ETH.
	
	The approximation guarantee lower bounds for \MHV~follow from the inapproximability of \textsc{Clique} and \textsc{Independent Set}, as both of these problems cannot be approximated within $\O(n^{1-\epsilon})$, unless $\P\neq \NP$ \cite{Zuckerman2006}.
	We use the same reduction from \textsc{Independent Set} to \MHV, and note that $|E(G')|=\O(n^2)$ and $h=\ell=|V(G)|=n$.
	This finishes the proof.
\end{proofO}

\begin{theoremO}\label{thm:mhe_eth_lb}
	\MHEfull~with $\ell=3$ cannot be solved in time $2^{o(n+m)}$, unless ETH fails.
\end{theoremO}
\begin{proofO}
	In their work on multiterminal cuts \cite{Dahlhaus1994}, Dahlhahus et al.\ showed \NP-completeness of \textsc{3-Terminal Cut} (equivalently, \MHEfull~with $\ell=3$) by a linear reduction from the \textsc{Max Cut} problem.
	The \textsc{Max Cut} problem definition is given below.
	
	\begin{problemx}
		\problemtitle{\textsc{Max Cut}}
		\probleminput{A graph $G$ and an integer $k$.}
		\problemquestion{Can vertices of $G$ be partitioned into two sets $V(G)=V_1 \sqcup V_2$, so that the number of edges between $V_1$ and $V_2$ in $G$ is at least $k$, i.e.\ $|E(V_1, V_2)|\ge k$?}
	\end{problemx}
	
	The reduction they give is linear, so it is sufficient to prove that \textsc{Max Cut} cannot be solved in $2^{o(n+m)}$ time, unless ETH fails.
	Although this result may be well-known, we have not found any explicit statement about that \textsc{Max Cut} cannot be solved in subexponential time.
	For completeness, we state it here.
	
	\begin{lemma}\label{lemma:max_cut_eth}
		\textsc{Max Cut} cannot be solved in time $2^{o(n+m)}$, unless ETH fails.
	\end{lemma}
	\begin{proof}
		This result can be obtained by following the series of classical reductions by Papadimitriou and Yannakakis in \cite{Papadimitriou1991}.
		They reduce an instance of \textsc{Max 3-SAT} to an instance of \textsc{Max 3-SAT} that contains at most three occurences of each variable, then to \textsc{Independent Set}, \textsc{Max 2-SAT}, then to \textsc{Max Not-All-Equal 3-SAT} and finally to \textsc{Max Cut}.
		All reductions they provide are linear, so no problem in this chain can be solved in subexponential time.
	\end{proof}

	We finish the proof by combining Lemma \ref{lemma:max_cut_eth} and the linear reduction in \cite{Dahlhaus1994} from \textsc{Max Cut} to \textsc{3-Terminal Cut}.
\end{proofO}
	\section{Algorithms}

In this section, we present two algorithms solving \MHV~or \MHE.
We start with a randomized algorithm for \MHV~that runs in $\Ostar{\ell^k}$ time and recognizes a yes-instance and finds the required coloring with a constant probability.
The algorithm is based on the following lemma.

\begin{lemmaO}\label{lemma:precolored_high_prob}
	Let $(G, p)$ be a graph with precoloring, and $P=\bigcup\limits_{i=1}^\ell \mathcal{H}_i(G,p)$. Let $c$ be a coloring that yields the maximum possible number of happy vertices in $(G,p)$, and let $H=\mathcal{H}(G,c)$ be the set of these vertices. Then $|H \cap P|\ge \frac{1}{\ell}\cdot |P|$.
\end{lemmaO}
\begin{proofO}
Let $k$ be the maximum possible number of vertices that can be happy simultaneously in $(G,p)$, so $(G,p,k)$ is a yes-instance of \MHV, and $(G,p,k+1)$ is not. Let $i$ be such that $|\mathcal{H}_t(G,p)|$ is maximum possible. In particular, $|\mathcal{H}_t(G,p)|\ge\frac{1}{\ell}\cdot |P|$. Construct a coloring $c'$ of $(G,p)$ by trivially extending $p$ with the color $t$. Note that all vertices that are not happy in $G$ with respect to $c'$ are contained in $\bigcup\limits_{i \in [\ell]\setminus \{t\}} \mathcal{H}_i(G,p)=P\setminus\mathcal{H}_t(G,p)$. 

Let $U=\mathcal{H}(G,p)\setminus P$. All vertices in $U$ and all vertices in $\mathcal{H}_t(G,p)$ are happy in $G$ with respect to $c'$. Let $H'=\mathcal{H}(G, c')$ be the set of all vertices that are happy in $G$ with respect to $c'$. Then $U \cap \mathcal{H}_t(G,p)\subseteq H'$, and $|H'|=|H'\cap U|+|H'\cap P|\ge |U|+\frac{1}{\ell}\cdot |P|$.

Take now a coloring $c$ that yields the maximum possible number of happy vertices in $(G,p)$, and let $H=\mathcal{H}(G,c)$ be the set of these vertices. In particular, $|H|\ge |H'|$. Suppose $|H \cap P| < \frac{1}{\ell}\cdot |P|$. But $|H|=|H\cap U| + |H\cap P|< |U|+\frac{1}{\ell} \cdot |P|\le|H'|$. This contradiction finishes the proof.
\end{proofO}

\begin{theoremO}\label{thm:mhv_lk_randomized_exact}
	There is a $\Ostar{\ell^k}$ running time randomized algorithm for \MHVfull.
\end{theoremO}
\begin{proofO}
Firstly, we provide a procedure that finds an answer for a given instance $(G,p,k)$ of $\ell$-\MHV~with a probability of at least $\ell^{-k}$, if $(G,p,k)$ is a yes-instance. The procedure is given in Fig.\ \ref{fig:algo_guess_answer}.

\begin{figure}[!ht]
	\centering
	\begin{algorithm}[H]
		\LinesNumbered
		\TitleOfAlgo{$\texttt{guess\_answer}(G, p, k)$}
		\KwIn{An instance $(G,p,k)$ of $\ell$-\MHV.}
		\KwOut{A set $H \subseteq \mathcal{H}(G, p)$ such that all vertices in $H$ can be happy in $(G, p)$ simultaneously.}
		\BlankLine
		\SetAlgoVlined
		\DontPrintSemicolon
		$P \longleftarrow \bigcup\limits_{i=1}^\ell \mathcal{H}_i(G, p)$\;
		$U \longleftarrow \mathcal{H}(G,p)\setminus P$\;
		\If{$k \le |U|$}{
			\Return $U$\nllabel{line:guess_answer_u}\;
		}
		\If{$P=\emptyset$}{
			\Return $\emptyset$\nllabel{line:guess_answer_empty}\;
		}
		
		$v \longleftarrow $ random vertex in $P$, each with equal probability\nllabel{line:guess_answer_guess_v}\;
		$p' \longleftarrow \left.p\right|_{V(G)\setminus \{v\}}$\;
		$i \longleftarrow $ the color such that $v \in \mathcal{H}_i(G,p)$\;
		\ForEach{$u \in N(v)$}{
			$p'(u) \longleftarrow i$\;
		}
		\Return $\texttt{guess\_answer}(G\setminus \{v\}, p', k-1) \cup \{v\}$\nllabel{line:guess_answer_induction}\;
	\end{algorithm}
	\caption{A randomized procedure finding a set of vertices that can be happy simultaneously.}\label{fig:algo_guess_answer}
\end{figure}

\begin{claim}\label{claim:mhv_randomized_correct}
	$\texttt{guess\_answer}(G,p,k)$ always outputs a set of vertices that can be happy simultaneously in $(G,p)$.
\end{claim}
\begin{claimproof}
	We prove that by induction on $k$. For $k=0$, the procedure returns $U$ in line \ref{line:guess_answer_u}, and all vertices in this set can be happy simultaneously in $(G,p)$ (for example, with any trivial extension of $p$). Let now $k>0$ and the claim statement hold for $k-1$. Consider in which line the procedure returns for $(G,p,k)$. If it is line \ref{line:guess_answer_u} or line \ref{line:guess_answer_empty}, then the claim statement holds. Consider the only case left, when procedure returns in line \ref{line:guess_answer_induction}.
	
	Let $H'$ be the set returned by $\texttt{guess\_answer}(G\setminus\{v\}, p', k-1)$. By induction, all vertices in $H'$ can be happy simultaneously in $(G\setminus\{v\}, p')$. Note that all vertices in $H'\cup\{v\}$ can be happy simultaneously in $(G,p)$ as well, since $p'$ is just a restriction of $p$ that ensures that all neighbours of $v$ are colored with the same color as $v$ itself. The claim statement follows immediately.
\end{claimproof}

\begin{claim}
	For any yes-instance $(G,p,k)$ of $\ell$-\MHV, the $\texttt{guess\_answer}$ procedure outputs a set $H$ with $|H|\ge k$ with a probability of at least $\ell^{-k}$. 
\end{claim}
\begin{claimproof}
	The proof is by induction on $k$. For $k=0$, the claim statement holds, since any output suffices. Let now $k>0$ and let the claim statement hold true for $k-1$. Finally, let $(G,p,k)$ be a yes-instance of \MHV. Consider how the procedure processes the instance. If it returns in line \ref{line:guess_answer_u}, having that $|U|\ge k$, it outputs a required set with the probability of $1$. Note that the procedure can't return in line \ref{line:guess_answer_empty}, since $(G,p,k)$ is a yes-instance. Hence, the only case left is that the procedure returns in line \ref{line:guess_answer_induction}. We now consider this case.
	
	Fix any optimal coloring $c$ of $(G,p)$. In particular, $c$ yields at least $k$ happy vertices in $(G,p)$, so $|\mathcal{H}(G,c)|\ge k$. By Lemma \ref{lemma:precolored_high_prob}, $|\mathcal{H}(G,c) \cap P| \ge \frac{1}{\ell}\cdot |P|$. Hence, in line \ref{line:guess_answer_guess_v}, the procedure chooses $v$ such that $v \in \mathcal{H}(G, c)$, with a probability of at least $\ell^{-1}$. Consider the case that the procedure indeed chooses such $v$, so $v \in \mathcal{H}(G,c)$. Since $v$ is precolored and is happy with respect to $c$, for any neighbour $u$ of $v$ holds $c(u)=c(v)=p(v)$. Thus, $\left.c\right|_{V(G)\setminus\{v\}}$ is a coloring extending $p'$ in $G \setminus \{v\}$. Moreover, $\mathcal{H}(G\setminus \{v\}, \left.c\right|_{V(G)\setminus\{v\}})=\mathcal{H}(G,c)\setminus \{v\}$. Therefore, $(G\setminus\{v\}, p', k-1)$ is a yes-instance.  By induction, $\texttt{guess\_answer}(G\setminus\{v\}, p', k-1)$ returns a set $H'$ with $|H'|\ge k-1$ and all vertices in $H'$ can be happy simultaneously in $(G\setminus \{v\}, p')$, with a probability of at least $\ell^{-k+1}$. By Claim \ref{claim:mhv_randomized_correct}, all vertices in $H'\cup\{v\}$ can be happy in $(G,p)$ simultaneously. Recall that $v \in \mathcal{H}(G,c)$ with a probability of at least $\ell^{-1}$, and obtain the total probability of at least $\ell^{-k}$.
\end{claimproof}

From the claims above immediately follows that a single launch of the procedure finds that the given instance is a yes-instance with a probability of at least $\ell^{-k}$, and never finds that if the given instance is a no-instance. We finish the construction of the randomized algorithm by saying that it repeats the procedure for $\ell^k$ times for the given instance, so it recognizes a yes-instance with a constant probability of at least $e^{-1}$, and never recognizes a no-instance as a yes-instance.
\end{proofO}

Note that by Theorem \ref{thm:mhv_clique_eth}, no algorithm with running time $\O(\ell^{o(k)})$ exists for \MHV, unless ETH fails.
Similarly, no $\O(\ell^{o(k)})$ running time \emph{randomized} algorithm exists for \MHV~under the \emph{randomized} ETH \cite{Dell2014}.
The algorithm given above is optimal in that sence.

We now turn onto \MHE~and give an exact algorithm with $\O^*(2^k)$ running time for this problem.
In its turn, this algorithm optimal in a sence that no $2^{o(k)}$ running time algorithm exists for \MHE~under ETH (see Theorem \ref{thm:mhe_eth_lb}).
The algorithm relies on the following kernelization result.
We note that this kernelization result and an algorithm with the running time of $\O^*(2^k)$ was already presented by Aravind et al.\ in \cite{Aravind2017}.
We believe that our kernelization algorithm is short and somewhat simpler, since it relies on a single reduction rule.

\begin{theoremO}[\cite{Aravind2017}]\label{thm:mhe_k_uncolored}
	\MHEfull~admits a kernel with at most $k$ uncolored vertices.
\end{theoremO}
\begin{proofO}
	Let $(G, p, k)$ be an instance of \MHE. We show how to obtain an equivalent instance $(G', p', k')$ of \MHE, where the number of uncolored vertices in $(G', p')$ is at most $k'$ and $k' \le k$.
	
	The kernelization algorithm consists just of applying the following reduction rule exhaustively to $(G, p, k)$.
	
	\begin{rrule}\label{rule:mhe_k_vertices}
		If there is a connected component $C$ consisting only of uncolored vertices in $(G, p)$, remove it from $G$ and reduce $k$ by the number of edges in $C$. That is, replace instance $(G, p, k)$ with an instance $(G \setminus C, \left.p\right|_{V(G)\setminus C}, k-|E(G[C])|)$.
	\end{rrule}

	The correctness of the reduction rule follows from the fact that one can color a connected component of uncolored vertices with the same single color and make all edges in the component happy.

	\begin{claim}
		After the exhaustive application of Reduction rule \ref{rule:mhe_k_vertices}, if $G'$ contains at least $k'$ uncolored vertices, then $(G', p', k')$ is a yes-instance. 
	\end{claim}
	\begin{claimproof}
		If $k' \le 0$, $(G', p', k')$ is trivially a yes-instance. Suppose now $k' \ge 1$ and $G'$ contains at least $k'$ uncolored vertices. We construct a coloring of $(G', p')$ that yields at least $k'$ happy edges.
		
		Take any uncolored vertex in $(G', p')$, say $v$, such that $v$ has at least one precolored neighbour in $(G',p')$. Note that such choice of $v$ always exists, otherwise Reduction rule \ref{rule:mhe_k_vertices} would be applied.
		
		Let $u$ be a precolored neighbour of $v$. Then, set the color of $v$ to $p(u)$. Edge $uv$ becomes happy, and still no connected component in $G'$ consists only of uncolored vertices. Thus, we can take an uncolored vertex with a precolored neighbour again. Repeat this procedure until no uncolored vertex remains in $G'$. The procedure is repeated at least $k'$ times, and each time a happy edge is obtained, so $(G', p', k')$ is a yes-instance.
	\end{claimproof}

	The statement of the theorem follows directly from the claim.
\end{proofO}

\begin{theoremO}[\cite{Aravind2017}]\label{cor:mhe_k_exact}
	There is a $\Ostar{2^k}$ running time algorithm for \MHEfull.
\end{theoremO}
\begin{proofO}
	As shown by Aravind et al.\ in \cite{Aravind2017}, \MHE~can be solved in time $\Ostar{2^{n'}}$, where $n'$ is the number of uncolored vertices. By Theorem \ref{thm:mhe_k_uncolored}, we can assume that $n' \le k$, and the statement follows.
\end{proofO}

	\bibliography{ref}
\end{document}